\documentclass[accepted]{uai2026} 
                        
\usepackage[american]{babel}

\usepackage{natbib} 
\usepackage{mathtools} 
\usepackage{booktabs} 
\usepackage{tikz} 

\usepackage{amssymb}
\usepackage{amsmath}
\usepackage{amsthm}
\usepackage{dsfont}\usepackage{bm}
\usepackage{url}
\usepackage{placeins}
\usepackage{float}

\hypersetup{
colorlinks=true,
citecolor=blue,
linkcolor=blue,
urlcolor=blue
}

\usepackage[capitalize]{cleveref}
\usepackage[ruled,vlined]{algorithm2e}

\usepackage{shadethm}
\usepackage{etoolbox}

\newtheoremstyle{new}
  {12pt}      
  {12pt}      
  {\itshape}  
  {}          
  {\bfseries\color{black}} 
  {.}         
  { }         
  {}          
\theoremstyle{new}
\newtheorem{theorem}{Theorem}[section]
\newtheorem{corollary}[theorem]{Corollary}

\newtheorem{lemma}[theorem]{Lemma}
\newtheorem{definition}[theorem]{Definition}

\newtheorem{example}[theorem]{Example}
\newtheorem{remark}[theorem]{Remark}

\AfterEndEnvironment{theorem}{\noindent\ignorespaces}
\AfterEndEnvironment{corollary}{\noindent\ignorespaces}
\AfterEndEnvironment{proposition}{\noindent\ignorespaces}
\AfterEndEnvironment{lemma}{\noindent\ignorespaces}
\AfterEndEnvironment{definition}{\noindent\ignorespaces}
\AfterEndEnvironment{assumption}{\noindent\ignorespaces}
\AfterEndEnvironment{example}{\noindent\ignorespaces}
\AfterEndEnvironment{remark}{\noindent\ignorespaces}
\AfterEndEnvironment{claim}{\noindent\ignorespaces}

\Crefname{theorem}{Theorem}{Theorems}
\Crefname{corollary}{Corollary}{Corollaries}
\Crefname{proposition}{Proposition}{Propositions}
\Crefname{lemma}{Lemma}{Lemmas}
\Crefname{definition}{Definition}{Definitions}
\Crefname{example}{Example}{Examples}
\Crefname{remark}{Remark}{Remarks}
\Crefname{claim}{Claim}{Claims}

\usepackage[breakable, theorems, skins]{tcolorbox}
\tcbset{enhanced}

\definecolor{shadethmcolor}{cmyk}{0,0,0,0.075}    
\definecolor{shaderulecolor}{rgb}{1,1,1}   

\newtheoremstyle{shad}
  {12pt}      
  {12pt}      
  {\itshape }  
  {}          
  {\bfseries\color{black}} 
  {.}         
  { }         
  {}          
\theoremstyle{shad} 

\newshadetheorem{thmbox}[theorem]{Theorem}
\newshadetheorem{corbox}[theorem]{Corollary}
\newshadetheorem{propbox}[theorem]{Proposition}
\newshadetheorem{lembox}[theorem]{Lemma}
\newshadetheorem{exbox}[theorem]{Example}
\newshadetheorem{defbox}[theorem]{Definition}
\newshadetheorem{rembox}[theorem]{Remark}

\newtheoremstyle{shad*}
  {12pt}      
  {12pt}      
  {\itshape }  
  {}          
  {\bfseries\color{black}} 
  {.}         
  { }         
  {}          
\theoremstyle{shad*} 
\newshadetheorem{anobox}{}

\Crefname{theorem}{Theorem}{Theorems}
\Crefname{corbox}{Korollar}{Corollaries}
\Crefname{propbox}{Proposition}{Propositions}
\Crefname{lembox}{Lemma}{Lemmas}
\Crefname{exbox}{Beispiel}{Examples}
\Crefname{defbox}{Definition}{Definitions}
\Crefname{rembox}{Anmerkung}{Remarks}

\usepackage{bm}
\usepackage{dsfont}

\newcommand{\ind}{\mathds{1}}
\newcommand{\R}{\mathds{R}}
\newcommand{\N}{\mathds{N}}

\newcommand{\E}{\mathds{E}}

\providecommand{\G}{}
\renewcommand{\G}{\mathds{G}}

\newcommand{\Acal}{\mathcal{A}}

\newcommand{\Ncal}{\mathcal{N}}

\newcommand{\Tcal}{\mathcal{T}}

\newcommand{\Wcal}{\mathcal{W}}

\newcommand{\eps}{\varepsilon}

\renewcommand{\bar}{\overline}

\providecommand{\Pr}{}
\renewcommand{\Pr}{\mathbb{P}}
\newcommand{\var}{{\mathds{V}\mathrm{ar}}}

\newcommand{\cov}{{\mathds{C}\mathrm{ov}}}
\newcommand{\corr}{{\mathrm{Corr}}}

\newcommand{\wh}[1]{\widehat{#1}}

\makeatletter
\def\blfootnote{\gdef\@thefnmark{}\@footnotetext}
\makeatother

\title{Online Bootstrap Inference for the Trend of Nonstationary Time Series}

\author[1,2]{Thomas Nagler}
\author[1,2]{Tobias Brock}
\author[1,2]{Nicolai Palm}
\affil[1]{%
    Department of Statistics\\
    LMU Munich\\
    Munich, Germany
}
\affil[2]{%
    Munich Center for Machine Learning\\
    Munich, Germany
}

\begin{document}
\maketitle

\begin{abstract}
This article proposes an online bootstrap scheme for nonparametric level estimation in nonstationary time series. Our approach applies to a broad class of level estimators expressible as weighted sample averages over time windows, including exponential smoothing methods and moving averages. The bootstrap procedure is motivated by asymptotic arguments and provides well-calibrated uniform-in-time coverage, enabling scalable uncertainty quantification in streaming or large-scale time-series settings. This makes the method suitable for tasks such as adaptive anomaly detection, online monitoring, or streaming A/B testing. Simulation studies demonstrate good finite-sample performance of our method across a range of nonstationary scenarios. In summary, this offers a practical resampling framework that complements online trend estimation with reliable statistical inference.
\end{abstract}

\section{INTRODUCTION}


\subsection{Motivation}

Real-time decisions often hinge on the evolution of a time-varying trend in a data stream. 
Such trends are commonly estimated by exponentially weighted moving averages (EWMAs) and related smoothers, which enable efficient online computation and adapt naturally to changing conditions. 
These estimators are widely used in applications ranging from monitoring athlete training-load metrics \citep{murray2017calculating}, to technical analysis of financial assets via \emph{exponential Bollinger bands} \citep{BollingerRules}, to industrial process control charts for mean-shift detection \citep{roberts2000control,LucasSaccucci1990}. 
Despite their ubiquity for estimation, principled \emph{inference} for EWMA-type smoothers---confidence intervals and sequential tests---remains underdeveloped. 
Existing approaches either assume stationarity of the noise or recompute batch estimators on sliding windows, both of which are poorly suited to streaming settings with nonstationarity, latency constraints, and memory budgets.

\subsection{Contribution}
We propose \emph{online bootstrap} procedures that provide scalable uncertainty quantification for EWMA-type estimators of the level (or trend) of nonstationary time series, including higher-order smoothers. 
In particular, our algorithm (\cref{sec:algorithm}) delivers real-time pointwise and uniform-in-time confidence intervals (CIs) and tests with:
\begin{enumerate}[label=(\roman*)]
    \item $O(1)$ time and memory cost per observation,
    \item locally adaptive interval widths and critical values, and
    \item approximate validity under serial dependence and nonstationarity.
\end{enumerate}
The latter is substantiated through mathematical arguments in an idealized setting (\cref{sec:asymptotics}) and numerical experiments (\cref{sec:experiments}).
Taken together, our results establish a theoretically grounded and practically effective resampling framework that complements online trend estimation with reliable inference.

\subsection{Related Work}

\paragraph{Confidence sequences/sequential monitoring}
Recent work on asymptotic \emph{confidence sequences} (CS) provides time-uniform inference via supermartingale constructions, yielding confidence bounds that hold under continuous monitoring \citep{howard2021time,WaudbySmithArbourSinhaKennedyRamdas2024TimeUniformCLTACS,Johari2022}. These methods primarily target cumulative sums, typically in \textit{iid}\ or martingale difference settings. They are not designed for inference on weighted smoothers under serial dependence, where direct martingale arguments fail.
More classical sequential procedures from process monitoring, including CUSUM charts and related change-detection methods \citep[e.g.,][]{lorden1971procedures, tartakovsky2014sequential}, provide time-uniform false-alarm control but are tailored to detecting abrupt changes. However, they usually make parametric or independence assumptions and do not directly yield calibrated confidence bands for smoothly evolving trends.

\paragraph{Time Series Bootstrap}
For dependent data, block bootstraps \citep{Kunsch1989, politis1992circular} and sieve bootstraps \citep{Buhlmann1997} provide valid inference but are inherently offline and require tuning choices. 
Multiplier methods avoid explicit blocking \citep{Shao2010}, yet standard implementations recompute batch statistics. 
Bootstrap approaches for nonparametric trend bands \citep{FriedrichSmeekesUrbain2020} likewise remain computationally demanding and are not designed for streaming settings.
Our work differs by providing \emph{streaming} resampling methods specialized to EWMA-type smoothers, with constant per-observation update cost. 
The closest in spirit is the autoregressive online bootstrap of \citet{PalmNagler2024} for stationary streams. Our method extends theirs to nonstationary environments, uniform-in-time inference, and non-Gaussian bootstrap multipliers.  


\paragraph{Adaptive Conformal Inference} Recent developments in adaptive conformal prediction approaches for sequential or streaming data \citep{shafer2008tutorial,Tibshirani2019,Gibbs2021,zaffran2022adaptive} share our goal of real-time inference, but for a different target. Conformal methods provide inference for label $Y_t$, which is itself a random variable with non-vanishing variance. We target a smoothed trend, a deterministic target, where the variance approaches zero as the size of the smoothing window increases.
\section{BACKGROUND}

\subsection{Setup}

In what follows, we assume that $(X_t)_{t \in \N}$ is a generic real-valued time series with instantaneous mean $\mu(t) = \E[X_t]$, $t \in \N$. 
As a running example, we can think of this series as heart rates recorded by an optical sensor in a smartwatch. The measurements have two sources of fluctuation. The true heart rate $\mu(t)$ fluctuates because of varying health conditions and activities, and due to the measurement error $X_t - \mu(t)$ stemming from the limited-frequency optical sensor.

Because we may only ever see a single observation $X_t$ with mean $\mu(t)$, there is little hope to estimate $\mu$ without making strong assumptions. The typical remedy is to focus on smoothed means 
\begin{align} \label{eq:smoothed-means}
    \mu_\eta(t) = \sum_{i = 1}^t  w_{t, \eta}(i) \mu(i),
\end{align}
and corresponding estimators
\begin{align} \label{eq:smoothed-estimator}
    \wh \mu_\eta(t) = \sum_{i = 1}^t  w_{t, \eta}(i) X_i.
\end{align}
Here, $\eta \in \R^p$ is a parameter controlling the amount of smoothing (or, conversely, forgetting). 


\subsection{Online Estimators of the Trend}

In an online learning or streaming data setting, the most common trend estimators are exponential weight smoothers  \citep{GARDNER2006637, hyndman2008forecasting}.
We give a few popular examples in the following.

\begin{example}[EWMA]\label{ex:ewma}
For an initial $s_0=0$, update
\begin{align*}
    s_t = \eta X_t + (1-\eta)s_{t-1}, 
\end{align*}
where $\eta \in \mathbb{R}$ and set $\wh \mu_\eta(t) = s_t$. This can be written as \eqref{eq:smoothed-estimator} with 
\begin{align*}
    w_{t, \eta}(i) = \eta(1-\eta)^{t-i}.
\end{align*}
\end{example}

\begin{example}[Brown double exponential smoothing]
For initial $s^{(1)}_0=s^{(2)}_0=0$ and $\eta \in \mathbb{R}$, update
\begin{align*}
    s^{(1)}_t &= \eta X_t + (1-\eta)s^{(1)}_{t-1},\\
    s^{(2)}_t &= \eta s^{(1)}_t + (1-\eta)s^{(2)}_{t-1}.
\end{align*}
We take the trend estimate as $\wh \mu_{\eta}(t) := 2s^{(1)}_t - s^{(2)}_t$. This can be written as \eqref{eq:smoothed-estimator} with weights
\[
  w_{t,\eta}(i) = \eta \bigl[2 - \eta (t-i+1)\bigr] (1-\eta)^{t-i}.
\]
\end{example}

\begin{example}[Triple exponential smoothing with additive seasonality]
For seasonal period $L$ and initial values \(s_0 = b_0 = c_{-L} = \dots = c_{0}=0\), the trend update is
\[
  \wh \mu_\eta(t) = s_t = \eta_1(X_t - c_{t-L}) + (1-\eta_1)(s_{t-1} + b_{t-1}),
\]
with trend estimate
\[
  b_{t} = \eta_2(s_t - s_{t-1}) + (1-\eta_2)b_{t-1},
\]
and additive seasonality update
\[
  c_t = \eta_3(X_t - s_t) + (1-\eta_3)\,c_{t-L},
\]
where $\eta \in [0,1]^3$. Closed forms of $w_{t,\eta}(i)$ can be derived for the above expressions but are omitted for simplicity.
\end{example}


\subsection{The Effective Sample Size}

For computing an unweighted mean, $\bar X_n = \frac{1}{n} \sum_{i = 1}^n X_i$, the sample size $n$ is the key scaling factor for the uncertainty. The situation is different for exponential smoothers and their variants discussed in the previous section. To see this, consider the case where $X_1, X_2, \ldots$ are i.i.d.
Comparing 
\begin{align*}
    \var[\bar X_n] &= \var[X_1] / n, \\
    \var[\wh \mu_\eta(n)] &=  \var[X_1] \left(\sum_{i = 1}^n w_{n, \eta}(i)^2\right),
\end{align*}
for large $n$, we may interpret the number 
\begin{align*}
    \nu_\eta = \left(\sum_{i = 1}^n w_{n, \eta}(i)^2\right)^{-1},
\end{align*}
as the \emph{effective sample size} of $\wh \mu_\eta$, a formula often attributed to \citet{kish1992weighting}.  For example, the standard EWMA weights have $\nu_\eta \approx (2 - \eta) / \eta$. Importantly, the effective sample size depends mainly only on the smoothing parameter $\eta$, not on the number of raw samples the smoother has processed.

For nonstationary time series, the above reasoning becomes murkier because $\var[\wh \mu_\eta(n)]$ involves potentially changing variances and covariances across time points. 
However, the quantity $1 / \nu_\eta$ still provides the right scaling for the variance of $\wh \mu_\eta$ under standard assumptions on the serial dependence. We will thus use the effective sample size $\nu_\eta$ throughout the paper to quantify the effective number of samples that are averaged by $\wh \mu_\eta$.

\section{PROBLEM FORMULATION}


Suppose we observe $X_1, \ldots, X_{t_0}$ from a nonstationary time series and begin computing the online estimators $\wh \mu_\eta(t)$ as new data arrive. 
Our goal is to perform statistical inference on the corresponding smoothed means $\mu_\eta(t)$ at a given significance level $\alpha$, either by constructing confidence intervals or by testing for changes in the underlying trend. 
For instance, in the heart-rate example, this corresponds to building real-time confidence bands for the instantaneous heart rate or detecting abnormal shifts in activity patterns.


\subsection{Uniform-in-time inference}
Classical inference typically provides pointwise guarantees, e.g., coverage for a single target.
In streaming contexts, however, many time points are monitored successively, and pointwise intervals will be exceeded frequently.
To guard against this, we seek procedures with (approximate) \emph{uniform-in-time validity} over a contiguous time interval $\Tcal = \{t_1 + 1, \dots, t_2\}$.

\paragraph{Confidence Bands}
A $(1-\alpha)$-confidence band is a collection of random intervals 
\begin{align*}
    \{[\wh \mu_\eta(t) - \wh c_{t, \alpha}, \wh \mu_\eta(t) + \wh c_{t, \alpha}]\colon \, t \in \Tcal\},
\end{align*}
with $\wh c_{t, \alpha}$ thresholds such that 
\begin{align*}
    \Pr\left( \exists t \in \Tcal \colon |\wh \mu_{\eta}(t) - \mu_{\eta}(t)| > \wh c_{t, \alpha}\right) \le \alpha.
\end{align*}
The thresholds $\wh c_{t,\alpha}$ adapt to the local level of uncertainty and are designed to be updated in real time, using only $O(1)$ computation and memory. 
Uniform validity means that the probability of a band miss in $\Tcal$ is approximately $\alpha$. This is a rather unforgiving property requiring additional caution.

\paragraph{Testing}
We consider tests for the null hypothesis 
\begin{align*}
    H_0\colon \quad \mu_\eta(t) = 0 \quad \forall \, t \in \Tcal,
\end{align*}
using the rejection rule
\begin{align*}
    \text{Reject if} \quad  \exists\, t \in \Tcal \colon \quad |\wh \mu_\eta(t)| > \wh c_{t, \alpha},
\end{align*}
with the same requirements on the thresholds $\wh c_{t, \alpha}.$
The testing framework is more general than it may appear. 
First, one-sided tests follow directly by dropping the absolute value in the rejection rule. 
Second, by choosing suitable weights $w_{t,\eta}(i)$ or applying the smoother to a transformed series, we can target a wide range of hypotheses beyond simple level shifts.

\begin{example}[Jump detection]
Suppose we want to detect sudden increases in heart rate. 
We define a `jump' as a change of size at least $J$ within a short horizon of $h$ time steps (e.g., 10 seconds). 
This leads to the null hypothesis
\[
  H_0\colon \quad \left| \sum_{i=h+1}^t w_{t,\eta}(i)\,[\mu(i)-\mu(i-h)] \right| \le J, \quad \forall\, t \in \Tcal.
\]
\end{example}

Constructing valid tests amounts to finding a confidence interval under the null. We shall therefore restrict much of the following exposition to confidence intervals for simplicity.

\subsection{Desiderata}

In summary, we have the following desiderata for our inference procedures:
\begin{itemize}
    \item[(N)] \textbf{Nonstationary}: The CIs and tests adapt to the local uncertainty in a nonstationary time series.
    \item[(U)] \textbf{Uniform validity}: The procedures give approximately valid inference uniformly in time.
    \item[(S)] \textbf{Small-window robustness}: Procedures remain valid even when effective sample sizes are small.
    \item[(O)] \textbf{Online}: Updates require only constant time and memory per new observation.
\end{itemize}


\section{METHOD} \label{sec:algorithm}

In the following, we describe how to construct confidence thresholds $\wh c_{t,\alpha}$ that satisfy the desiderata (N)–(O). 
The construction proceeds in three steps: (i) an idealized asymptotic design, (ii) a bootstrap approximation that makes the procedure practical, and (iii) small-sample calibrations to ensure robustness in practice. A formal result proving asymptotic validity of the procedure is given in the following section.

\subsection{Idea}

\paragraph{Idealized Setting}
To motivate our construction, consider first an idealized scenario where the distribution of the process is known. 
The idea is to scale thresholds with the standard error $\sigma_t=\var[\wh\mu_\eta(t)]^{1/2}$ of the smoother and then inflate them by a factor $q_\alpha$ to achieve uniform-in-time coverage. 
To extend validity over long horizons, we partition the monitoring period into blocks of increasing length and use past blocks to calibrate the intervals for the next one.
Assume for clarity that $t_2=t_0+2(t_1-t_0)$ so that the periods $[t_0\!+\!1,t_1]$ and $[t_1\!+\!1,t_2]$ have the same length.

For coverage over a period $\Tcal$, we calibrate $q_\alpha$ such that
\begin{align*}
    \Pr\left(\sup_{t\in\Tcal}\frac{|\wh\mu_\eta(t)-\mu_\eta(t)|}{\sigma_t}>q_\alpha\right)\le \alpha.
\end{align*}
For longer horizons, let 
$K=\big\lceil \log_2\!\big((t_2-t_0)/(t_1-t_0)\big)\big\rceil$ and partition 
$\Tcal=\bigcup_{k=1}^K \Tcal_k$, where
\begin{align*}
    \Tcal_k&=\{t_{1}^{(k)}\!+\!1,\dots,t_{2}^{(k)}\},\\
t_{1}^{(k)}&=t_0+2^{k-1}(t_1-t_0),\\
t_{2}^{(k)}&=t_0+2^{k}(t_1-t_0).
\end{align*}
Calibrating $q_{k,\alpha}$ so that
\begin{align*}
    \Pr\!\left(\sup_{t\in\Tcal_k}\frac{|\wh\mu_\eta(t)-\mu_\eta(t)|}{\sigma_t}>q_{k,\alpha}\right)\le \frac{\alpha}{K}
\end{align*}
and then using a union bound yields the overall level $\alpha$.
Other stitching/spending schedules can be used instead of the constant $\alpha/K$ budget for every period \citep[e.g.,][]{howard2021time}. These may also allow for $K \to \infty$, but we shall not pursue this further in this already technical article.

\paragraph{Bootstrap Approximation}
In practice, $\mu_\eta(t)$ and $\sigma_t$ are unknown, so we replace the idealized error process with a bootstrap analogue, to be defined formally in \Cref{sec:implementation}.
The bootstrap generates approximate replicates of the estimation error, from which both scale ($\sigma_t^*$) and quantiles ($q_{k,\alpha}^*$) can be estimated empirically. Suppose, for the moment, we have access to an independent, approximate version $\wh \delta_{\eta}^{*}(t)$ of the estimation error $\wh \delta_\eta(t) = \wh \mu_{\eta}(t) - \mu_\eta(t)$ with known variance $\sigma_t^{*2} = \var[\wh \delta_{\eta}^{*}(t)]$.
Then we may as well define the thresholds as $c_{t, \alpha}^* = \sigma_t^* q_{k, \alpha}^*$, where
\begin{align*}
    \Pr\left(\sup_{t \in \Tcal_k} |\wh \delta_{\eta}^{*}(t)| /  \sigma_t^* > q_{k, \alpha}^* \right)
     & \le \frac{\alpha}{K}.
\end{align*}
A remaining issue is that an estimate of the quantile $q_{k, \alpha}^*$ must be available ahead of period $\Tcal_k$, so we cannot use data from that interval. If the fluctuations around the nonstationary trend are sufficiently stable in the sense that
\begin{align*}
     \Pr\left(\sup_{t \in \Tcal_k} \frac{|\wh  \delta_\eta(t)| }{\sigma_t^*} > y \right)   
     & \approx  \Pr\left(\sup_{t \in [t_0 + 1, t_1^{(k)}]}  \frac{|\wh  \delta_\eta(t)| }{\sigma_t^*} > y \right),
\end{align*}
we can simply choose $y = q_{k, \alpha}^*$ such that the right-hand side is bounded by $\alpha / K$. Because the interval lengths are identical on both sides, this certainly holds when the noise is stationary, but may only hold approximately otherwise.

\paragraph{Small-sample Calibration}
Our asymptotic arguments leading to \Cref{thm:true} ahead rely on Gaussian approximations of the sequence of estimation errors $\wh \delta_\eta(t)$ and their bootstrapped versions $\wh \delta_\eta^*(t)$.
However, in practice, we may want to use small smoothing windows, where the asymptotic approximation is not accurate.
Since uniform coverage is especially sensitive to such deviations, we introduce corrections that stabilize variance estimates and make the procedure more robust in finite samples.

\subsection{Implementation} \label{sec:implementation}

Based on the ideas above, we propose the following concrete implementation.
Fix a smoothing parameter $\eta$, burn-in time $t_0$, start time $t_1$, and end time $t_2$. We devise an algorithm that computes the smoothed estimates $\wh \mu_{\eta}(t)$ and estimated confidence intervals in an online manner.
The resulting procedure alternates between (i) generating bootstrap replicates of the smoothed error process, (ii) updating online variance estimates, and (iii) recalibrating critical values at block boundaries. 
A detailed version is provided in Algorithm~\ref{algo:main} and explained in more detail in the following paragraphs. The confidence intervals produced by the algorithm on a few simulated time series are illustrated in \Cref{fig:wide-compare-processes}. 

\begin{algorithm}[t]
    \textbf{Input}:  $\eta$, $\alpha$, $t_0, t_1, t_2$, $B_1, B_2$, and  $X_1, X_2, \ldots$\\

    \SetAlgoLined

    \textbf{Initialize} $\wh \mu_\eta=\wh \mu_\eta(t_0), \rho_\eta = 1 - \nu_\eta^{-\chi}, Z^{(b)} = \wh \delta_\eta^{*(b)}(0)= m_0^{(b)} = 0 , K = \lceil \log_2\big((t_2-t_0)/(t_1-t_0)\big) \rceil
$. \\[6pt]



    \For{$ t=t_0 + 1, \dots, t_2$}{ 


    \For{$b=1 \ to \ B$}{

        Simulate $\zeta^{(b)}\sim \mathcal{N}(0,1)$. \\
        Update \\[-18pt]
        \begin{align*}
            Z^{(b)}             & \leftarrow \ \rho_\eta Z^{(b)} +\sqrt{1-\rho^2_\eta}\zeta^{(b)} \\
            V^{(b)}             & \leftarrow \ T(Z^{(b)})                                       \\
            X^*                 & \leftarrow \ V^{(b)}(X_t-\wh{\mu}_{\eta})                \\
            \wh  \delta_\eta^{*(b)} & \leftarrow  \text{smoother update using }  X_t^*      \\[-24pt]
        \end{align*}
    }

    $\wh \sigma^{*2} \leftarrow \text{empirical variance of } \{\wh \delta_\eta^{*(1)}, \dots, \wh \delta_\eta^{*(B_1)} \}$.\\
    \For{$b=B_1 + 1 \ to \ B$}{
        $m^{(b)} \leftarrow \max\{m^{(b)}, |\wh \delta_\eta^{*(b)}| / \wh \sigma^{*} \}$. \\
    }
    \If{$t = t_0  + 2^k(t_1 - t_0)$ for some $k \in \N_0$}{
        $\wh q_{\alpha}^* \leftarrow$ empirical $(1 - \alpha / K)$-quantile of \\ \qquad \, $\{m^{(B_1 + 1)}, \dots, m^{(B)}\}$. \\
    } \, \\[-8pt]
    $\wh{\mu}_\eta \leftarrow \ \text{smoother update using } X_t$  \\
    Set confidence threshold $\wh c_{t, \alpha}^* = \wh q_{\alpha}^* \wh \sigma^{*}$.
    }
    \caption{Online nonstationary AR-bootstrap with uniform-in-time calibration}
    \label{algo:main}
\end{algorithm}

\paragraph{Bootstrap Multipliers}

Over the whole time horizon $t = t_0 + 1, t_0 + 2, \ldots$, we compute the smoothed estimates $\wh \mu_{\eta}(t)$ online using an appropriate update formula.
We further construct $B = B_1 + B_2$ independent bootstrap replicates of the estimation error $\wh \delta_{\eta}(t) = \wh \mu_{\eta}(t) - \mu_{\eta}(t)$. The first $B_1$ replicates will be used to estimate the variance, and the remaining replicates will be used for calibration. Our experiments suggest that $B_1 = 20, B_2 = 80$ are sufficient for good accuracy. 
To mimic the dependence structure of the data, we generate autoregressive Gaussian multipliers with persistence parameter $\rho_\eta = 1 - \nu_\eta^{-\chi}$, $\chi \in (0, 1/2)$. 
This follows the stationary online bootstrap of \citet{PalmNagler2024}, but we additionally apply a heavy-tailed transformation $T$ to make the procedure robust against non-Gaussian innovations when $\nu_\eta$ is small.

For each $b = 1, \ldots, B$, we generate a sequence of bootstrap multipliers $V_1^{(b)}, V_2^{(b)}, \ldots$ with
\begin{align*}
    V_{t}^{(b)} = T(Z_t^{(b)}), \quad Z_t^{(b)} = \rho_\eta Z_{t - 1}^{(b)} + \sqrt{1 - \rho_\eta^2} \xi_t^{(b)},
\end{align*}
where $\xi_1^{(b)}, \xi_2^{(b)}, \ldots \stackrel{iid}\sim \Ncal(0, 1)$, and $T = t_{2 + \nu_\eta^{1/3}}^{-1} \circ \Phi$ with $t_\nu$ the CDF of a $t$-distribution with $\nu$ degrees of freedom and $\Phi$ the CDF of a standard normal distribution.  We introduce the heavy-tailed transformation $T$ to make the approximation more robust against non-Gaussianity when $\nu_\eta$ is small.

\paragraph{Local Adaptation}

Now define the bootstrapped errors
\begin{align*}
    \wh \delta_{\eta}^{*(b)}(t) = \sum_{i = 1}^t w_{t, \eta }(i)  V_i^{(b)}  (X_i - \wh \mu_{\eta}(i - 1)).
\end{align*}
To capture rapid changes in the time series, we center each bootstrap innovation on the lagged smoother $\wh \mu_\eta(t-1)$. 
This makes the bootstrap variance more responsive to local changes, making the subsequent calibration step easier.
This is a major difference to the stationary online bootstrap of \citet{PalmNagler2024}, where the multipliers are applied to the mean-centered observations $X_i - \bar X_n$, which is not suitable for nonstationary data.
The bootstrapped errors can be computed online---using online updates for $\wh \mu_{\eta}(t)$, $V_t^{(b)}$, and noting that the error $\wh \delta_\eta^*(t)$ is just the initial smoother applied to the
sequence 
$X_t^* = V_t^{(b)} (X_t - \wh \mu_{\eta}(t - 1)).$ 
From here, we can estimate the squared standard error $\sigma_t^{*2} = \var[\wh \delta_{\eta}^{*}(t)]$ by the empirical variance   $\wh v_t$ of $\{\wh \delta_{\eta}^{*(b)}(t)\colon b = 1, \dots, B_1\}$. 


\paragraph{Global Calibration}

At block boundaries $t_1^{(k)}$, we recalibrate the critical value for the upcoming block to ensure that, by a union bound across blocks, the overall exceedance probability remains below $\alpha$.
Throughout the series, we compute the maximal observed standardized errors
\begin{align*}
    m_t^{(b)} = \max_{s \in [1, t]} |\wh \delta_{\eta}^{*(b)}(s)| / \wh \sigma_s^{*}, \quad b = B_1 + 1, \dots, B.
\end{align*}
and their empirical quantiles
\begin{align*}
    \wh q_{k, \alpha}^* = \inf\left\{q \colon \frac{1}{B_2} \sum_{b = B_1 + 1}^B \ind\left\{m_{t_1^{(k)}}^{(b)} \le q\right\} \ge 1 - \frac{\alpha}{K}\right\}.
\end{align*}

\paragraph{Runtime and Memory Complexity} The algorithm proceeds sequentially over times $t_0 + 1, \dots, t_2$, iteratively updating a fixed number of objects. More precisely, the algorithm keeps track of the objects listed in the initialization phase, which is $O(B)$ in memory. At every time step, updates cost $O(1)$ for each bootstrap sample $b = 1, \dots, B$ and an additional $O(B)$ for computing the variance of bootstrap replicates. Finally, at $K = O(\log t_2)$ block boundaries, quantiles are computed via sorting, incurring $O(B \log B)$ cost. The overall  complexity is thus $O(B)$ memory and $O(t_2 B + \log(t_2) B \log B)$ computations up to time $t_2$.

\paragraph{Hyperparameter choices.}
The parameter $K$ controls the number of calibration blocks over the
monitoring horizon. For fixed $t_0$ and $t_2$, a larger $K$ leads to earlier and more frequent recalibration. This can improve adaptivity when the local
uncertainty changes abruptly, but it also allocates only $\alpha/K$
error probability to each block, and can therefore increase conservatism
when the distribution of the standardized error process is stable. In
practice, more stable series favor fewer, longer calibration
blocks.

The parameter $\chi \in (0,1/2)$ controls the persistence of the multiplier bootstrap weights through $\rho_\eta = 1-\nu_\eta^{-\chi}$. It does not affect the convergence rate of the estimator, but determines the range over which temporal dependence from the observations is propagated into the bootstrapped estimate. 
Our default choice $\chi=1/3$ worked well across the simulations, and the sensitivity experiments in \cref{sec:app_ablation} indicate that performance is not very sensitive to this choice.

Finally, a larger number of bootstrap replicates $B$ is generally preferred, although our experiments in  \cref{sec:app_ablation} suggest that $B\approx 100$ is already sufficient.

\section{Asymptotics} \label{sec:asymptotics}

Several design choices in the algorithm are motivated by asymptotic considerations.
The core idea is to approximate the sequence of weighted sample averages by a Gaussian process, which becomes accurate as the effective sample size $\nu_\eta$ increases.
Moreover, we assume that the bootstrap quantile $q_{k, \alpha}^*$ is known exactly, which ignores Monte Carlo error and the effect of recalibration boundaries.
In this idealized asymptotic sense, we can formally prove the validity of our approach.
We stress that this argument is primarily motivational; the practical, non-asymptotic performance is assessed through numerical experiments in \Cref{sec:experiments}.


\subsection{Assumptions}

Intuitively, $\nu_\eta$ acts as the relevant time scale, so our asymptotic results are derived in a regime where $\nu_\eta \to \infty$.
We specify the burn-in, start, and end-times as 
$$t_0 = c_0 \nu_\eta, \quad t_1 = c_1 \nu_\eta, \quad t_2 = c_2 \nu_\eta,$$
 where $c_0, c_1, c_2 \in \N$ and $c_0 < c_1 < c_2$.
In this parametrization, $c_0$ is the number of length-$\nu_\eta$ periods we reserve for burn-in, $c_1 - c_0$ is the number of periods for the initial calibration, and $c_2 - c_1$ the number of periods we want to maintain valid inferences on. 

Our main result requires several assumptions on the data-generating process and the weights.
Suppose that there is $\gamma > 2$ such that:
\begin{enumerate}[label=(S\arabic*), leftmargin=0.9cm]
        \item \label{A1:X} $\sup_{i}\E[|X_i|^\gamma]<\infty$.
        \item \label{A2:X}$\sup_n \!\max_{m\leq n}m^{b}\beta_n(m)<\infty$ for $b>\gamma/\chi(\gamma-2).$\footnote{Here, $\beta_n(m)$ is the $\beta$-mixing coefficient of $X_1,\dots,X_n$, see \Cref{sec:assumptions} for a precise definition.}
\end{enumerate}
The first condition is a tail condition and requires the existence of a moment slightly higher than the second. The second condition is a restriction on the serial dependence, which requires that events become increasingly independent as they become more separated in time. The dependence has to fade more quickly if the tails of $X_i$ are heavy (i.e., $\gamma$ is small). These conditions are standard in the literature on time series analysis and relatively mild.

To state our conditions on the weights, we introduce the following pseudo-metric on $[0, 1]$:
$$d^w_\eta(s,t)=\left(\sum_{i=1}^{\infty}| w_{\lfloor sc_2\nu_{\eta}\rfloor, \eta}(i)- w_{\lfloor tc_2\nu_{\eta}\rfloor, \eta}(i)|^{\gamma}\right)^{\!\! 1/\gamma}\!\!.$$
Now assume:
\begin{enumerate}[label=(W\arabic*), leftmargin=1cm]
    \item \label{A0:weights} $\sup_{t,i,\eta}|\nu_{\eta}w_{t, \eta}(i)|<\infty$. 
    \item \label{A5:weights} There is  $C  \in [1, \infty)$ such that for all $s,t\in [1/c_2,1]$,
    $$d_\eta^w(s, t) \le C (| s - t| + \nu_\eta^{-1})^{1/C} < \infty.$$
\end{enumerate}
The first prevents the weights from spiking too heavily, while the second requires the weights to be sufficiently smooth in $s$. The conditions are mild and verified exemplarily for the EWMA weights from \Cref{ex:ewma} in \Cref{lem:weights}. 

\subsection{Main result}

\begin{theorem} \label{thm:true}
    Let $c_{t, \alpha}^* = \sigma_t^* q_{k, \alpha}^*$, $\sigma_t^{*2} = \var[\wh \delta_{\eta}^{*}(t)]$, and suppose assumptions \ref{A1:X}, \ref{A2:X}, \ref{A0:weights}, and \ref{A5:weights} hold.  If 
    \begin{enumerate}[label=(\roman*), leftmargin=0.9cm]
        \item $\mu(i) = \E[X_i]$ is constant, or
        \item $\lim_{\nu_\eta\to \infty}\max_{t\in (t_1,t_2]}\sigma_t/\sigma_t^*=0,$
    \end{enumerate} 
it holds that
        \begin{align*}
        \limsup_{\nu_\eta \to \infty}\Pr\left( \exists t \in (t_1, t_2] \colon |\wh \mu_{\eta}(t) - \mu_{\eta}(t)| > c_{t, \alpha}^*\right) \le \alpha.
    \end{align*}
\end{theorem}
The result covers two important cases: (i) when the trend is constant, and (ii) when the trend varies sufficiently strongly.
If the trend varies, it typically holds that $\sigma_t\ll\sigma_t^*$ as $\nu_\eta\to\infty$; a sufficient condition is given in \Cref{thm:2} in the appendix.
In this case, we obtain valid, but over-conservative bands. This is unavoidable under general nonstationarity; see \citet[Sec.~4]{palm2025centrallimittheoremsnonstationarity}. 
The main intuition is that the bootstrap variance $\sigma_t^{*2}$ captures the variability of the error process $\wh \delta_\eta(t)$, but is also affected by estimation bias under nonstationarity. 
In particular, one can show that
\begin{align*}
    \sigma_t^{*2} & \gtrsim \var\left[\sum_{i=1}^tw_{t,\eta}(i)V_i [\mu_\eta(i)-\mu(i)]\right].
\end{align*}
As a simple example, with linear trend $\mu(i) = a i$, EWMA weights,  independent multipliers $V_i$, and large $t$, this simplifies to
\begin{align*}
     \sigma_t^{*2} & \gtrsim \sum_{i=1}^t w_{t, \eta}^2(i) [\mu_\eta(i) - \mu(i)]^2 \approx  a^2 /2\eta.
\end{align*}
If $a \neq 0$,  $\sigma_t^{*2} $ diverges as  $\nu_\eta \approx 2/\eta \to \infty$, but $\sigma_t^2$ remains bounded, leading to conservative bands.

Intermediate regimes are more delicate: For dependent multipliers, $\sigma_t^{*2}$ may be smaller than $\sigma_t^2$ under specific interactions between the multiplier covariances and bias sequence $\mu_\eta(i) - \mu(i)$.
A complete characterization of these cases would require substantially more complicated assumptions on the trend and dependence structure. Nevertheless, our experimental results suggest that the proposed calibration remains reliable over a broad range of such finite-sample regimes.


\begin{figure*}[!t]  
  \centering
  \includegraphics[width=0.9\textwidth]{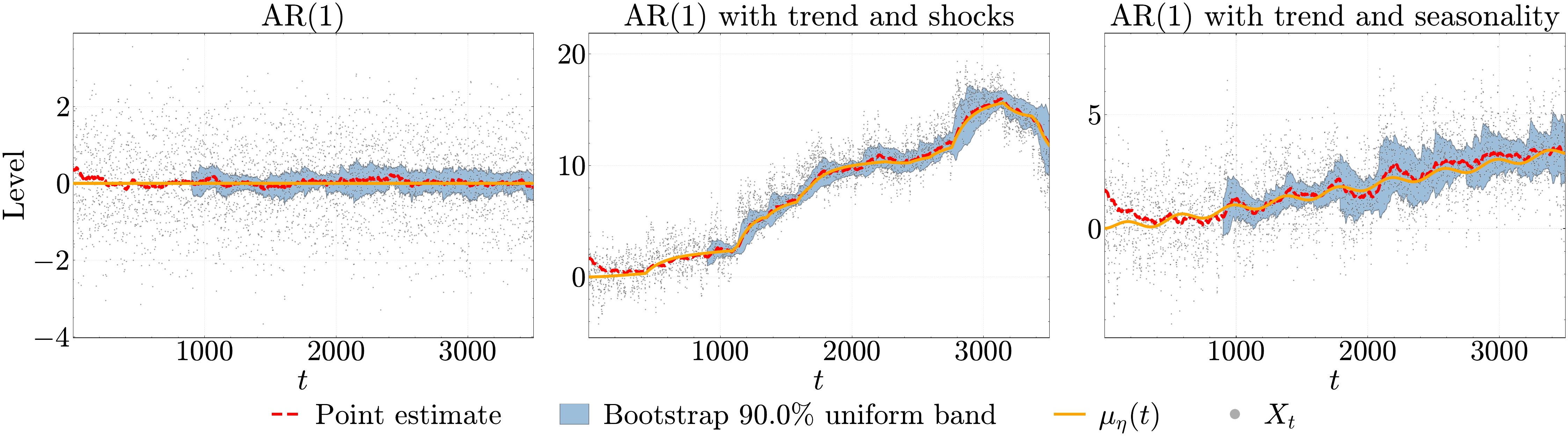}
  \caption{Simulated data, true and estimated smooth mean, and confidence bands obtained from Algorithm \ref{algo:main} with $\alpha=0.1$ and EWMA weights.}
  \label{fig:wide-compare-processes}
\end{figure*}
\section{EXPERIMENTS} \label{sec:experiments}

\begin{figure*}[t]
  \centering
  \includegraphics[width=1\textwidth,keepaspectratio]{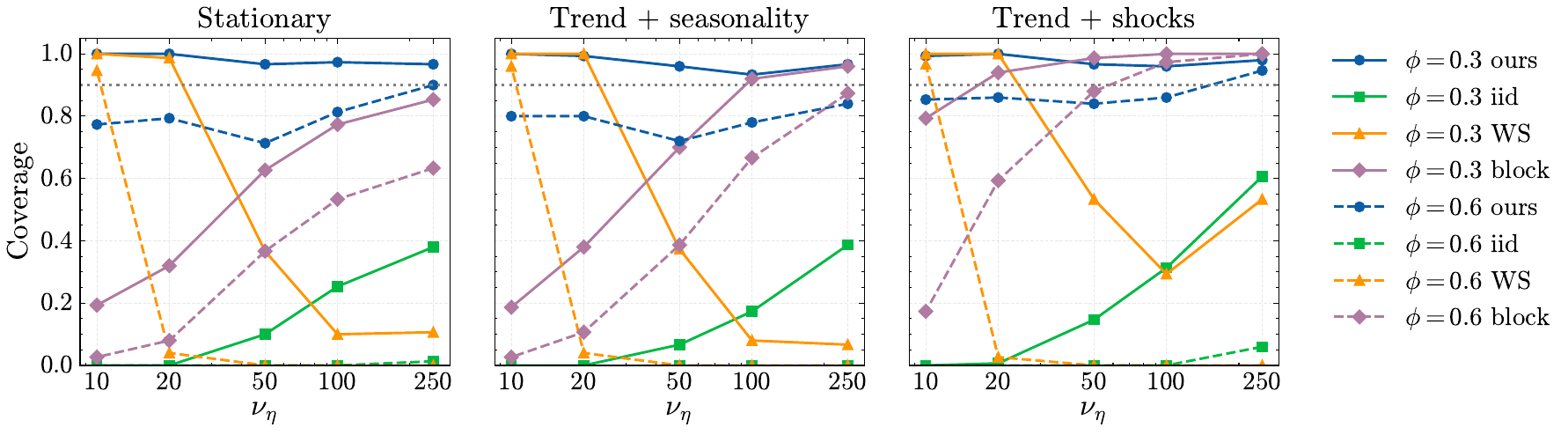}
  \caption{Uniform coverage for various DGPs and $\alpha=0.1$ using EWMA.}
  \label{fig:mainexp}
\end{figure*}

In this section, we assess the performance of Algorithm \ref{algo:main} through simulation experiments. In particular, we verify the theoretical results and the suggested finite-sample corrections, illustrating their necessity in nonstationary settings. 
All experiments can be reproduced using the accompanying repository\footnote{\url{https://github.com/Tobias-Brock/olbootstrapping}}.

\begin{figure*}[t] 
  \centering
  \includegraphics[width=1\textwidth,keepaspectratio]{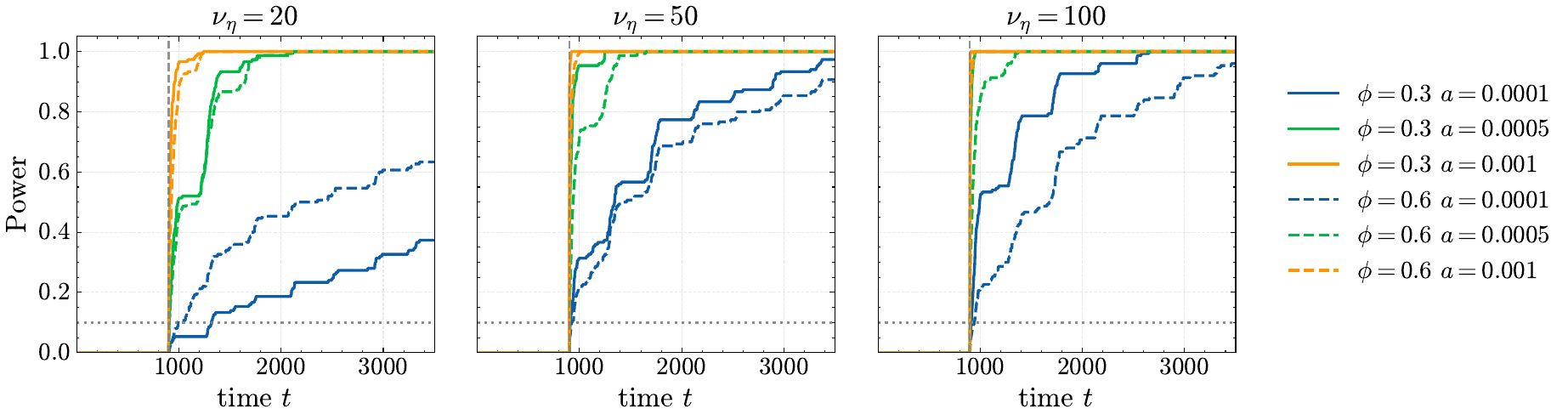}
  \caption{Power for detecting the non-constant trends in the Trend + seasonality DGP under varying trend slopes $a$ using EWMA; $a = 0.001$ is shown in \cref{fig:wide-compare-processes}.}
  \label{fig:power}
\end{figure*}

\subsection{Data Generating Process}


We simulate univariate time series $(X_i)_{i=1}^n$ with mean
\begin{align*}
   \E[X_i] =  m_i = \mu  + \tau_i + s_i+L_i,
\end{align*}
where $\mu \in \mathbb{R}$ is a constant level, $\tau_i$ is a linear trend, $s_i$ is a seasonal component and $L_i$ is a random shock level. The trend and seasonal component are
\begin{align*}
    \tau_i = a\cdot i, \quad s_i = A \sin\left(\frac{2\pi i}{P}+\psi\right)
\end{align*}
for a slope $a\in \mathbb{R}$ and seasonal amplitude $A \in [0,\infty)$, where $P\in(0,\infty)$ is the seasonal period and $\psi\in[0,2\pi)$ the seasonal phase. At each time point $i$, a random shock occurs independently with probability $p \in[0,1]$, with $B_i \sim \text{Bernoulli}(p)$. If a shock occurs, draw $J_i \sim N(0,\sigma_J^2)$ independent across $i$. The accumulated shock level  is
\begin{align*}
    L_0 = 0, \qquad L_i=L_{i-1}+B_iJ_i.
\end{align*}
We consider a first-order autoregressive (AR(1)):
\begin{align*}
    X_i = m_i + \phi(X_{i-1}-m_{i-1}) + \varepsilon_i,
\end{align*}
where $\phi \in \mathbb{R}$ and $\varepsilon_i=\sigma z_i$ with $z_i \overset{iid}{\sim} \mathcal{N}(0,1)$.


\paragraph{Setup} We set $\mu=0$ and unless noted otherwise, $z_i\stackrel{iid}{\sim}N(0,1)$,
$\varepsilon_i=\sigma z_i$ with $\sigma=1$. Three regimes are considered for $\phi\in \{0.3,0.6\}$:
\begin{enumerate}
  \item Stationary:\quad $a = 0 \implies \tau_i=0 \ \forall i, s_i= 0, p=0$.
  \item Trend + seasonality:\quad $a=10^{-3}$, and $A=0.4$, $P=400$, $\psi=0$, $p=0$.
  \item Trend + shocks:\quad With $a=10^{-3}$, $s_i = 0$,
        and shocks occurring with $p=0.005$, $\sigma_J=2$.
\end{enumerate}
Figure \ref{fig:wide-compare-processes} illustrates time series simulated from the DGPs and confidence bands computed by our method.

We compare four methods:
\begin{itemize}[leftmargin=0.5cm]
    \item \textbf{ours}: Algorithm \ref{algo:main} with $\chi=1/3$.
    \item \textbf{iid}: Algorithm \ref{algo:main} with $\chi=0$, which reduces to an $iid$ Gaussian multiplier bootstrap \citep[e.g.,][]{chernozhukov2013gaussian}.
    \item \textbf{WS}: The Lindeberg-type Gaussian mixture martingale asymptotic confidence sequence of \citet{WaudbySmithArbourSinhaKennedyRamdas2024TimeUniformCLTACS} adapted to smoothed means (see \cref{sec:AsympCS}).
    \item \textbf{block}: A circular moving block bootstrap \citep{politis1992circular} adapted to smoothed means (see Appendix \ref{app:block-bootstrap}). 
\end{itemize}
The Gaussian multiplier bootstrap and the asymptotic confidence sequences are included mainly for reference, but should not be expected to give valid inferences because they assume independent data streams. To the contrary, the block bootstrap captures dependence in the data but is computed offline, leading to $O(Bt)$ update cost at every time point $t$.
\paragraph{Evaluation} 
We generate $150$ independent series per configuration of length $n=3500$, with a burn-in of $500$ observations and a calibration period of $t_1-t_0=400$. 
We use EWMA smoothing and set $B=400$ bootstrap replicates with $B_1 = B/5$ for a target confidence level $1-\alpha=0.9$. 
Our primary metric is the \emph{uniform-in-time coverage}, defined as the fraction of Monte Carlo replications in which 
$\sup_{t\in\mathcal{T}} |\wh\mu_\eta(t)-\mu_\eta(t)| \le \wh c_{t,\alpha}$ on the evaluation horizon $\mathcal{T}$.
Additional experiments for Brown smoothing and different $\alpha$ are in Appendix \ref{sec:app_addexp}.


\subsection{Results}


\paragraph{Coverage}
 Figure \ref{fig:mainexp} shows the coverage for the different DGPs using EWMA smoothing. For our proposed method, the AR(1) process with $\phi=0.3$ exhibits overcoverage, while the AR(1) process with $\phi=0.6$ slightly undercovers for small effective sample sizes $\nu_\eta$. This highlights the natural trade-off between the effective sample size and the strength of dependence in the process, as well as the tendency of time-series bootstraps to underrepresent autocovariances when the effective sample size is small. Similar behavior can be observed for the experiments in Appendix \ref{sec:app_addexp}. In all cases, the iid and WS baselines suffer from severe undercoverage because they cannot account for temporal dependence. The block bootstrap achieves nominal coverage only for larger effective sample sizes, since it does not incorporate a small-sample correction. 
Practitioners should therefore be cautious when the dependence range is
large relative to the smoothing window, and in such cases use larger
effective sample sizes or interpret the resulting bands conservatively.

\paragraph{Width}

Figure \ref{fig:stationary_width} shows a decrease in the average interval width similar to $1/\sqrt{\nu_\eta}$. This is in line with theoretical expectations, as the variance is of order $1/\nu_\eta$ asymptotically. This further validates that the intervals are appropriately scaled, even when the uniform coverage approaches 1, as in the case of the AR(1) process with $\phi=0.3$, where the interval strongly overcovers. 
A similar behavior can be observed for other configurations, as shown in Appendix A.1. Naturally, the baseline intervals' widths are substantially smaller than our intervals, as they fail to provide accurate coverage. 

 \begin{figure}[!htbp]
  \centering
  \includegraphics[width=1\columnwidth]{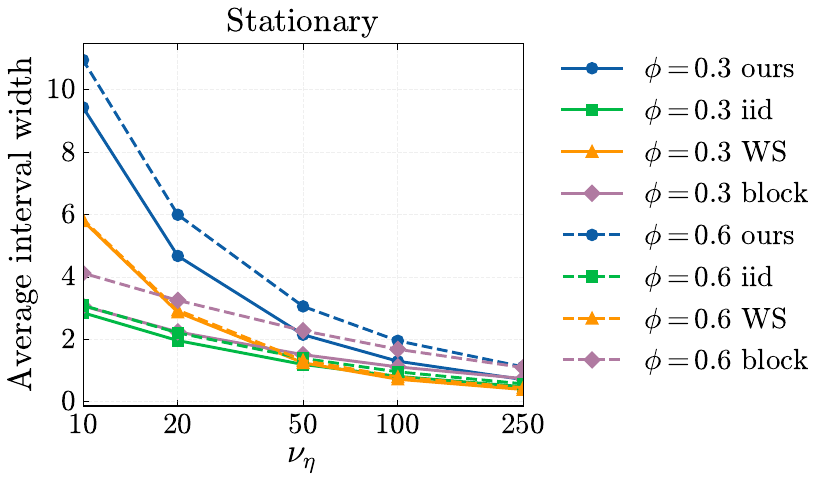}
  \caption{Average interval width for various stationary processes.}
  \label{fig:stationary_width}
\end{figure}







\paragraph{Power}
\cref{fig:power} illustrates the power of tests for the null hypothesis $H_0\colon \mu_\eta(t) = 0 \forall t$. The lines indicate the frequency of rejecting the null by time $t$ (starting after the burn-in period at $t = 900)$.
For easier reference, we simulate time series from the DGP shown in the right panel of \cref{fig:wide-compare-processes} (where $a = 0.001$). As expected, strong trends are detected very quickly, and the power deteriorates when the slope decreases, the effective sample size $\nu_\eta$ is small, or the serial dependence coefficient $\phi$ is large.

\paragraph{Runtime}
A per-update runtime comparison is illustrated in Figure~\ref{fig:rebuttal-runtime}. While the iid multiplier bootstrap and the WS baseline are slightly faster than our method, they fail to account for the dependence in the data. In contrast, the block bootstrap becomes increasingly expensive due to its $O(Bt)$ update cost. This illustrates the main computational advantage of our method: it retains dependence-aware inference while remaining suitable for streaming settings.

\begin{figure}[!htbp]
    \centering
    \includegraphics[width=1\columnwidth]{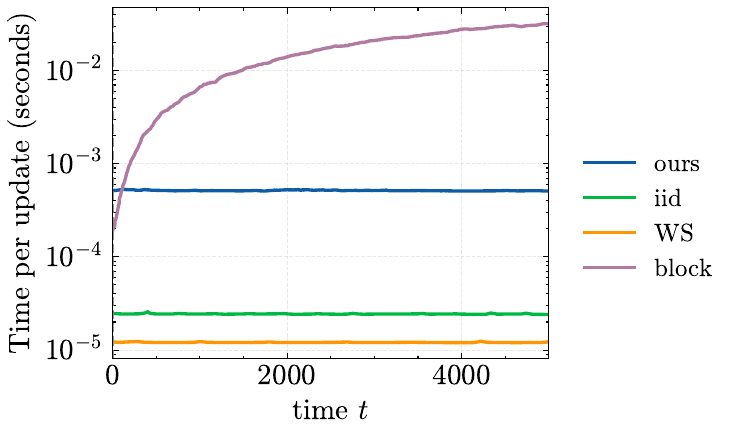}
    \caption{
    Per-update runtime comparison between methods.
    }
    \label{fig:rebuttal-runtime}
\end{figure}

\paragraph{Strongly nonstationary DGPs}
We further evaluate the methods on three strongly nonstationary DGPs that stress-test performance under nonlinear dependence, conditional heteroskedasticity, higher-order serial dependence, and structural breaks. The results in Appendix~\ref{app:strong-nonstationarity} show that our method remains robust across these more challenging settings, whereas the iid and WS baselines again suffer from undercoverage. Similarly, the block bootstrap underperforms for small effective sample sizes, while being significantly more expensive.

\paragraph{Sensitivity and ablation}
In Appendix \ref{sec:app_ablation}, we provide additional results on the sensitivity of our method to several design choices and hyperparameters; these results are summarized as follows.
As a preliminary stability check, we increase the series length to $n=5000$ and the number of replicates to 250, confirming that our method stabilizes with larger samples. We also conduct a set of ablation experiments. Varying the bootstrap size $B$ (keeping $B_1 = B/5$) shows that the method is stable for $B \geq 100$, whereas smaller values lead to noticeable deterioration. Changing the length of the calibration period $t_1 - t_0$ has a much weaker effect, with reliable results already for $t_1 - t_0 \geq 50$. Note that this also changes the number of recalibrations $K$, indicating robustness against this parameter as well. In contrast, the degrees of freedom parameter $\nu$ in the transformation $T$ strongly affects coverage for small effective sample sizes.  Setting $T = \mathrm{id}$ exacerbates the issue, underscoring the need for a finite-sample correction. Moreover, coverage remains stable for appropriate choices of $\chi$. Finally, replacing Gaussian noise with heavier-tailed $t_6$ innovations had only a minor impact. 



\section{CONCLUSION}

This article proposes an online bootstrap method for inference on nonstationary time series trends. It provides uniform-in-time confidence bands for EWMA-type smoothers with $O(1)$ update cost per observation and remains approximately valid under a broad range of conditions, including cases with very small effective sample sizes, trends, jumps, and strong dependence. This advances the toolbox for principled and scalable inference for real-time trend estimation in nonstationary data streams.

Looking ahead, an especially promising direction is integration with modern online learning techniques such as online gradient descent and mirror descent \citep{Hazan2016, cesa2021online}, adaptive regret minimization \citep{Jun2017}, and bandit-style algorithms \citep{Bubeck2012}. Such connections could facilitate a range of high-impact applications, from continuously monitored A/B testing \citep{Johari2022} to anomaly detection in streaming systems \citep{Ahmed2016}, as well as real-time analytics in finance, health, and other dynamic domains.

\bibliography{bibliography}

\newpage

\onecolumn

\title{Online Bootstrap Inference for the Trend of Nonstationary Time Series \\(Supplementary Material)}
\maketitle

\appendix

\section{ADDITIONAL EXPERIMENTS}
\label{sec:app_exp}
Section \ref{sec:app_addexp} extends our main experiments by testing the proposed method at a confidence level of $1-\alpha=0.8$ and under Brown smoothing. Section \ref{app:strong-nonstationarity} stress tests our method with strongly nonstationary processes, and Section \ref{sec:app_ablation} presents additional ablation results. 

\subsection{Results for Smoothing and Confidence Variations}
\label{sec:app_addexp}

Figures \ref{fig:mainexp_ewma_alpha02}, \ref{fig:mainexp_brown_alpha01}, \ref{fig:mainexp_brown_alpha02} illustrate that the method achieves consistent performance under varying $\alpha$ and Brown smoothing.

\begin{figure}[H]
  \centering
  \includegraphics[width=1\textwidth,keepaspectratio]{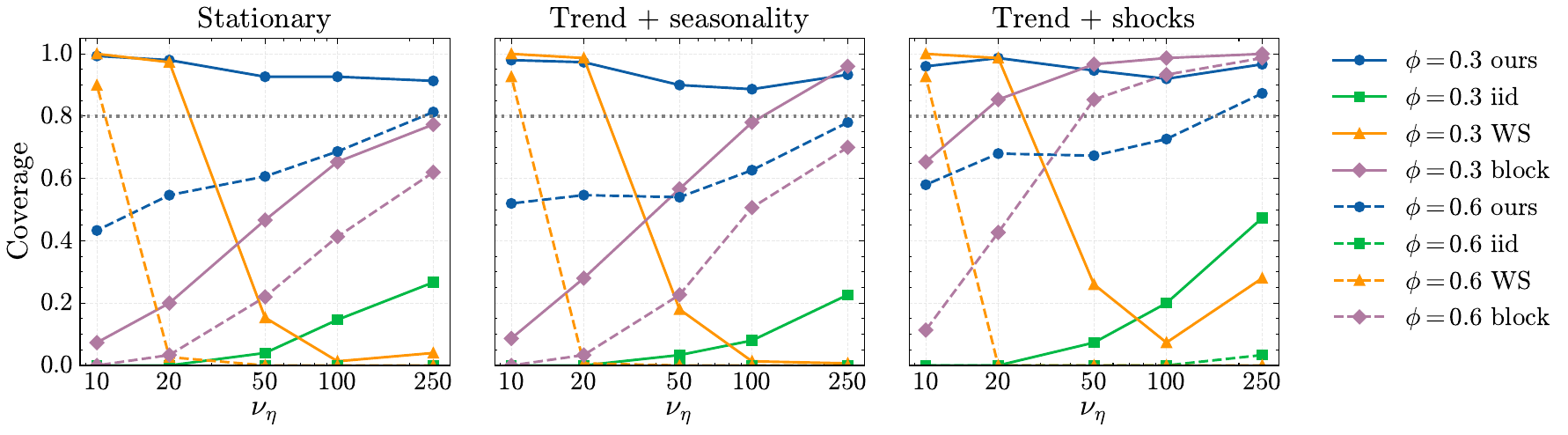}
  \caption{Uniform coverage for various DGPs and $\alpha=0.2$ using EWMA.}
  \label{fig:mainexp_ewma_alpha02}
\end{figure}

\FloatBarrier

\begin{figure}[H]
  \centering
  \includegraphics[width=1\textwidth,keepaspectratio]{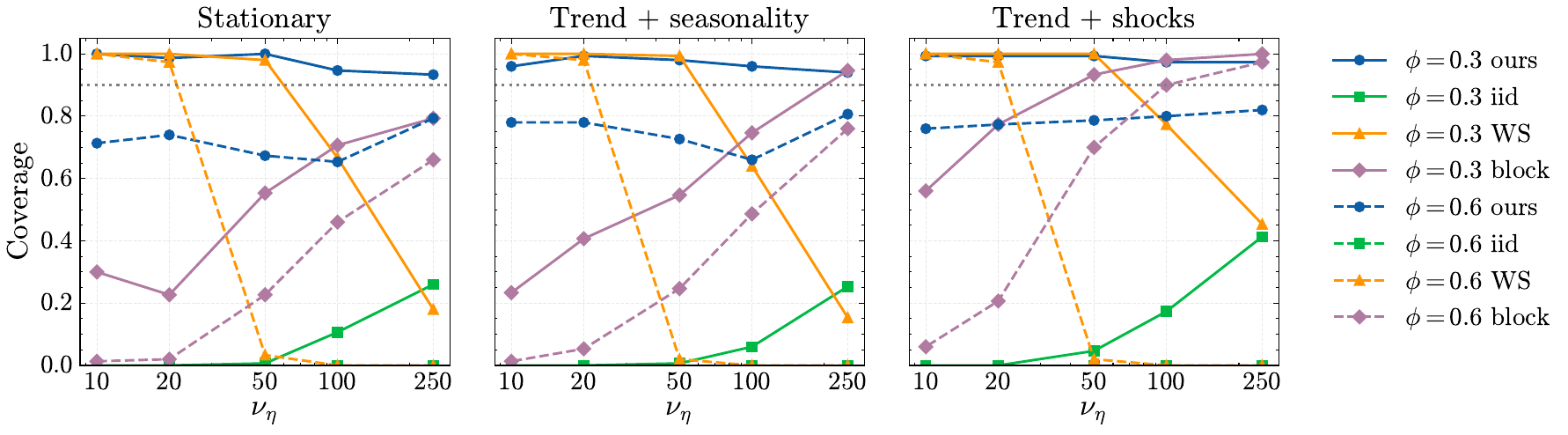}
  \caption{Uniform coverage for various DGPs and $\alpha=0.1$ using Brown.}
  \label{fig:mainexp_brown_alpha01}
\end{figure}

\FloatBarrier
\begin{figure}[H]
  \centering
  \includegraphics[width=1\textwidth,keepaspectratio]{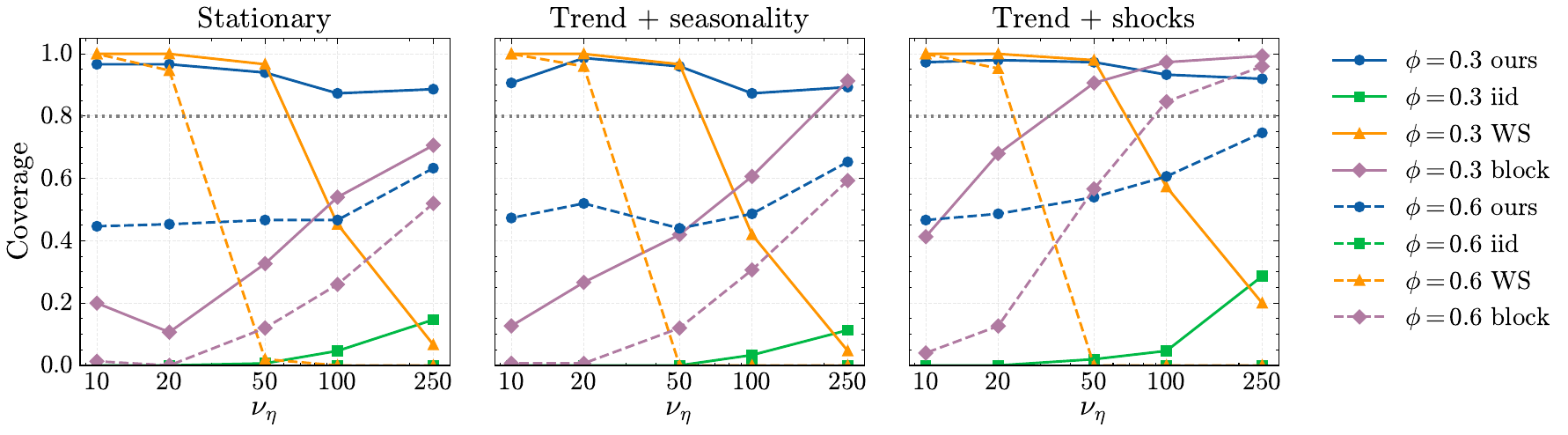}
  \caption{Uniform coverage for various DGPs and $\alpha=0.2$ using Brown.}
  \label{fig:mainexp_brown_alpha02}
\end{figure}


\subsection{Strong Nonstationarity}
\label{app:strong-nonstationarity}

We additionally evaluate the methods on more challenging data-generating processes with stronger forms of nonstationarity. These experiments are designed to stress-test performance under nonlinear dependence, conditional heteroskedasticity, higher-order serial dependence, and structural breaks.

We consider the following three regimes.

\paragraph{Sine-GARCH.}
This process combines a strongly seasonal mean with nonlinear and conditionally heteroskedastic fluctuations. The observed series is
\[
    X_i = m_i + Y_i,
    \qquad
    m_i = A \sin\left( \frac{2\pi i}{P} + \psi \right),
\]
where $A=1.5$, $P=400$, and $\psi=0$. The centered process follows
\[
    Y_i = c \sin(Y_{i-1}) + \sigma_i z_i,
\]
with nonlinear coefficient $c=0.8$ and $z_i \overset{iid}{\sim} N(0,1)$. The conditional variance evolves according to a GARCH(1,1) recursion
\[
    \sigma_i^2
    =
    \omega + \alpha Y_{i-1}^2 + \beta \sigma_{i-1}^2.
\]
We set $\alpha=0.08$, $\beta=0.90$, so that the persistence is $\alpha+\beta=0.98$, and choose $\omega=(1-0.98)$.

\paragraph{Sparse AR(20).}
A sparse higher-order autoregression with a quadratic trend and shocks:
\[
    X_i
    =
    m_i
    +
    \sum_{\ell=1}^{20} \phi_\ell
    \left(X_{i-\ell} - m_{i-\ell}\right)
    +
    \varepsilon_i,
    \qquad
    \varepsilon_i \overset{iid}{\sim} N(0,1).
\]
The lags $\ell \in \{1,2,5,10,20\}$ have nonzero coefficients $(0.1575, 0.1125, 0.09, 0.054, 0.036)$ with total mass $0.45$.
The mean contains a quadratic trend and permanent shocks,
\[
    m_i = a(i-1)^2 + L_i,
    \qquad
    a = 10^{-6},
\]
where
\[
    L_i = L_{i-1} + B_i J_i,
    \qquad
    B_i \sim \operatorname{Bernoulli}(0.01),
    \qquad
    J_i \sim N(0, 2^2).
\]

\paragraph{Structural breaks.}
The third process is an AR(1) process with piecewise linear trends, level jumps, and changes in the innovation scale, with $\phi = 0.45$. 
We sample two structural break points uniformly subject to a margin of 500 observations from the boundaries and a minimum distance of 600 observations between breaks. At each break, the sign of the local trend coefficient is reversed. The initial slope is $a=0.002$, and the two breaks additionally induce level jumps of size $3$ and $-3$, respectively. Within each segment, the mean is locally linear,
\[
    m_i = \ell_j + a_j \tau_i,
\]
where $\tau_i$ denotes the local time since the beginning of the current segment, $a_j \in \{a,-a\}$ is the segment-specific trend coefficient, and $\ell_j$ is the segment-specific level after accumulating the previous trend and jump effects. The innovation scale also changes with the local trend direction: we use standard deviation $0.7$ when the local trend coefficient is nonnegative and standard deviation $1.6$ when it is negative.

\paragraph{Evaluation.}
For all three regimes, we use the same experiment setup as in Section \ref{sec:experiments} with  EWMA smoothing and target confidence level $1-\alpha=0.9$ with $B=200$. Figure~\ref{fig:rebuttal-dgps} illustrates the DGPs with confidence bands, while Figure~\ref{fig:rebuttal-strong-nonstationarity} reports the corresponding uniform coverage.

\begin{figure}[ht]
    \centering
    \includegraphics[width=1\textwidth]{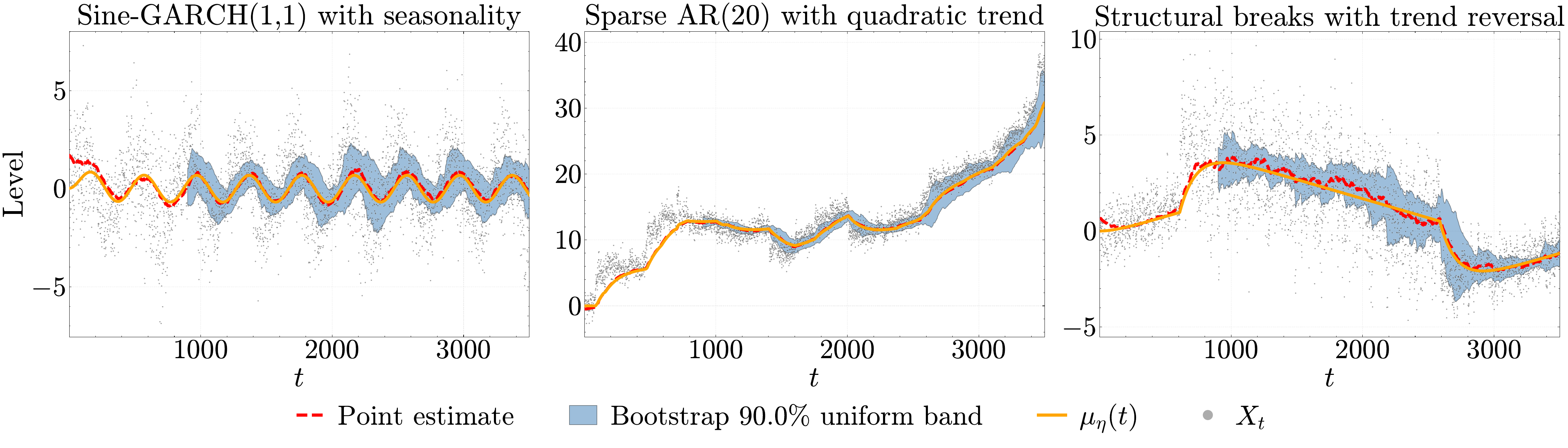}
    \caption{
     Strongly nonstationary data, true and estimated smooth mean, and confidence bands obtained from Algorithm 1 with $\alpha=0.1$ and EWMA weights.
    }
    \label{fig:rebuttal-dgps}
\end{figure}

\begin{figure}[ht]
    \centering
    \includegraphics[width=1\textwidth]{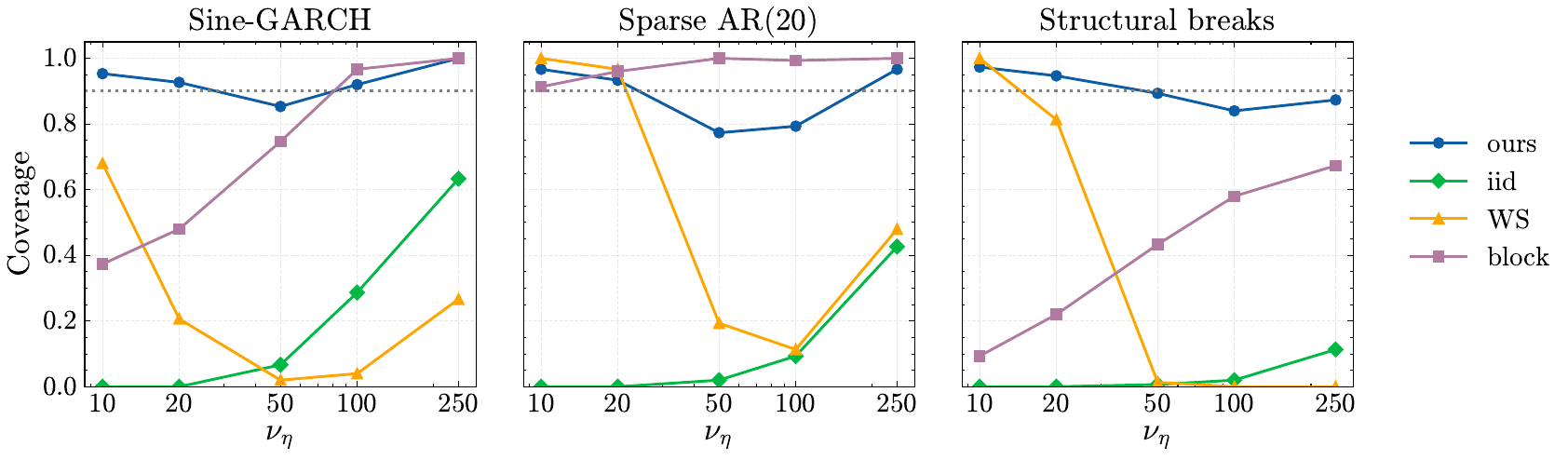}
    \caption{
    Uniform coverage for strongly nonstationary DGPs with $\alpha=0.1$ using EWMA.
    }
    \label{fig:rebuttal-strong-nonstationarity}
\end{figure}

\FloatBarrier
\subsection{Ablation Study}
\label{sec:app_ablation}

Figure \ref{fig:trends_width} shows that interval widths increase consistently also for nonstationary processes. Figure \ref{fig:ablationcompilation} further illustrates that removing the finite-sample correction $T$ drastically reduces coverage. Moreover, using heavier $t_6$ tails for the DGPs does not significantly affect performance. Finally, Figure \ref{fig:ar15000} demonstrates that our method stabilizes with large samples.

\begin{figure}[H]
  \centering
  \includegraphics[width=0.85\textwidth,keepaspectratio]{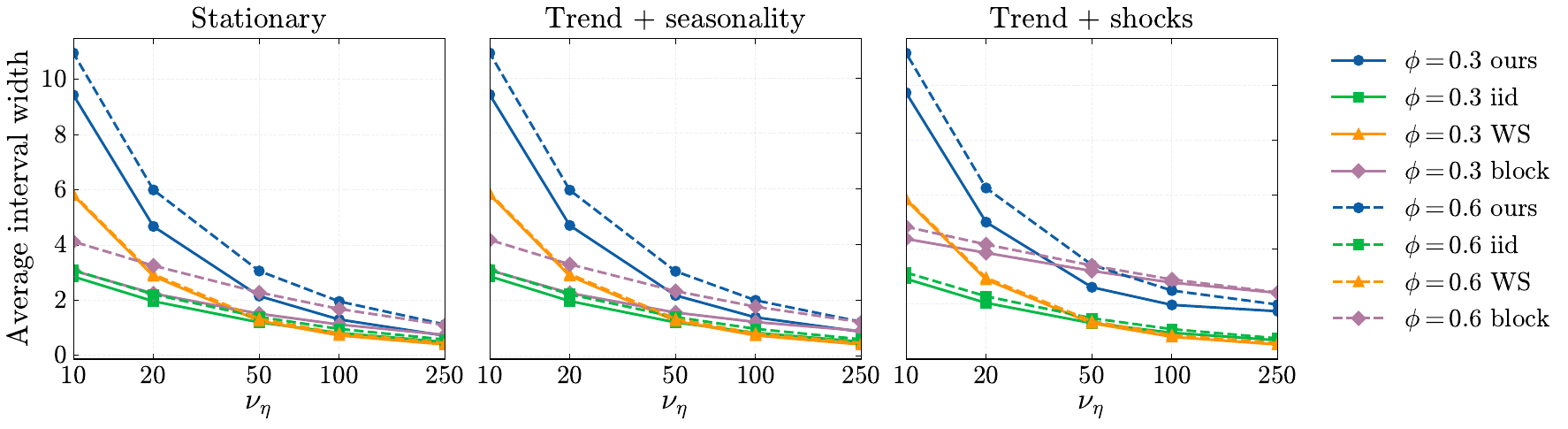}
  \caption{Average interval widths for various DGPs and $\alpha=0.1$ using EWMA.}
  \label{fig:trends_width}
\end{figure}
\FloatBarrier

\begin{figure}[H]   
  \centering
  \includegraphics[width=1\textwidth,keepaspectratio]{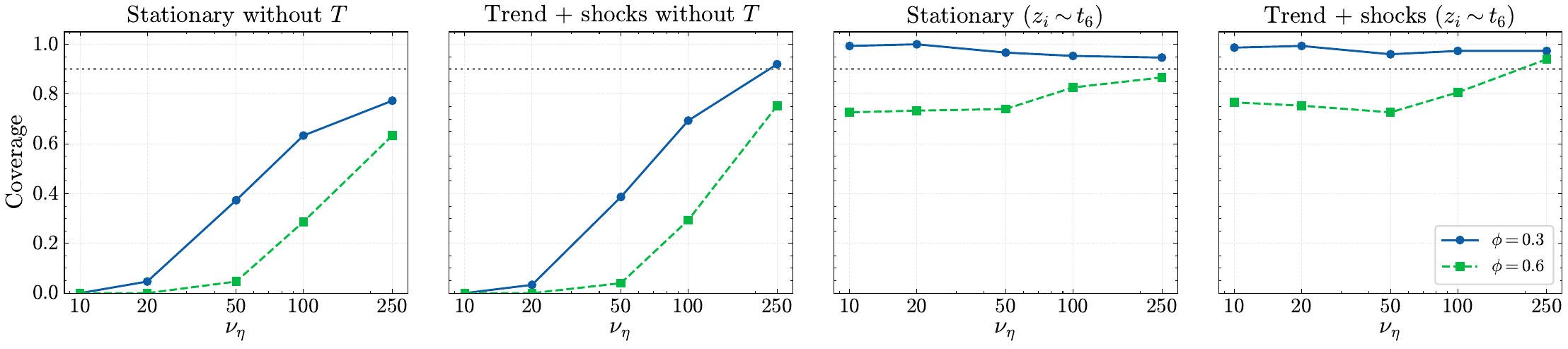}
  \caption{Uniform coverage of various AR(1) processes using EWMA without student transformation $T$ in bootstrap procedure (top) and with heavy $t_6$ tails (bottom).}
  \label{fig:ablationcompilation}
\end{figure}

\begin{figure}[H] 
  \centering
  \includegraphics[width=1\textwidth,keepaspectratio]{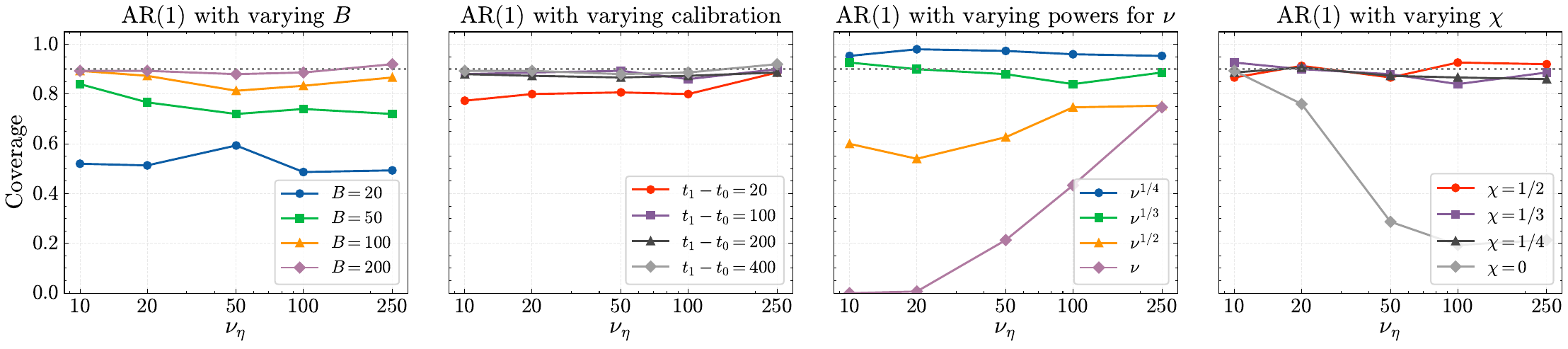}
  \caption{Uniform coverage for stationary AR(1) process with $\phi=0.5$ with varying values of $B$, $t_1-t_0$ and $\nu_\eta$.}
  \label{fig:ablation1}
\end{figure}

\begin{figure}[H]   
  \centering
  \includegraphics[width=0.5\textwidth,keepaspectratio]{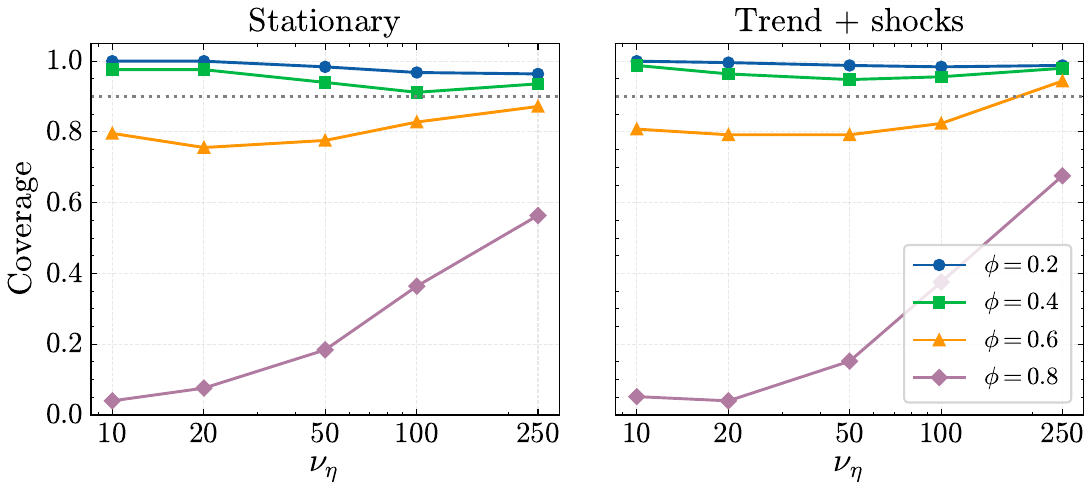}
  \caption{Uniform coverage of various AR(1) processes of size $n=5000$ for 250 time series using EWMA.}
  \label{fig:ar15000}
\end{figure}

\subsection{Asymptotic Confidence Sequences} \label{sec:AsympCS}

We adapt the AsympCS construction of \citet[Proposition 2.5]{WaudbySmithArbourSinhaKennedyRamdas2024TimeUniformCLTACS} to our smoothed mean target. Let $(X_t)_{t\ge1}$ denote the observed time series and let the smoother produce the online estimate $\widehat\mu_{\eta,t}$. A convenient way to connect AsympCS to this target is through one-step-ahead innovations. We define
\[
Y_t=X_t-\widehat\mu_{\eta,t-1},\qquad t\ge2,
\]
where $\widehat\mu_{\eta,t-1}$ is measurable with respect to $\sigma(X_1,\dots,X_{t-1})$. Writing the smoother in weight form,
\[
\widehat\mu_{\eta,t}=\sum_{i=1}^t w_{t,i}X_i,
\]
one can equivalently think of AsympCS as acting on the time-indexed sequence $Y_i^{(t)}=t w_{t,i}X_i$, whose average matches the smoothed target in expectation:
\[
\frac1t\sum_{i=1}^t\mathbb E[Y_i^{(t)}]=\sum_{i=1}^t w_{t,i}\mathbb E[X_i]=\mathbb E[\widehat\mu_{\eta,t}].
\]

In our experiments, we use the Gaussian-mixture AsympCS boundary with miscoverage level $\alpha$ and mixture parameter $\rho>0$, together with an online variance estimate $\widehat\sigma_t^2$. In the parameterization of \citet{WaudbySmithArbourSinhaKennedyRamdas2024TimeUniformCLTACS}, the resulting half-width is
\[
w_t=\sqrt{\frac{2(1+\nu_\eta\widehat\sigma_t^2\rho^2)}{\nu_\eta^2\rho^2}\log\left(\frac{\sqrt{1+\nu_\eta\widehat\sigma_t^2\rho^2}}{\alpha}\right)}.
\]
We report the anytime-valid confidence sequence for the smoothed mean level as
\[
\mathcal C_t=[\widehat\mu_{\eta,t}-w_t,\widehat\mu_{\eta,t}+w_t],\qquad t\ge1.
\]

\subsection{Block Bootstrap}
\label{app:block-bootstrap}

As an additional dependence-aware batch baseline, we implement a moving-block bootstrap for the smoothed mean estimator. At time $t$, let $\widehat \mu_\eta(i)$ denote the smoother applied to the observed series up to time $i$. We first form one-step-ahead innovations
\[
    r_i = X_i - \widehat \mu_\eta(i-1), \qquad i = 1,\ldots,t.
\]
This centering matches the residual construction used in our online bootstrap, while the subsequent resampling step follows a standard moving-block bootstrap. For each bootstrap replicate $b=1,\ldots,B$, we draw a circular moving-block bootstrap sample from the centered residuals and set
\[
    \ell = \left\lceil 4 \nu_\eta^{1/3} \right\rceil,
\]
which is the optimal block length up to constants \citep{politis2004automatic} and worked well in our experiments. Blocks of length $\ell$ are sampled uniformly from the circular residual sequence and concatenated until a bootstrap residual path $r_1^{*(b)},\ldots,r_t^{*(b)}$ of length $t$ is obtained. We then apply the same smoother, with the same initialization as for the original estimator.



Each recomputation resamples and smooths the complete history up to time $t$. Thus, a single recomputation costs $O(Bt)$ time and memory, and recomputing at every time point costs $O(Bt^2)$.

\section{THEORY AND PROOFS} \label{sec:proofs}

To simplify notation in the following, we shall assume that $c_1 = 1$ and denote $c_2 = c$.

\subsection{Assumptions} \label{sec:assumptions}

Recall the construction of the bootstrap weights:
\begin{align*}
    V_{t} = T(Z_t), \quad Z_t = \rho_\eta Z_{t - 1} + \sqrt{1 - \rho_\eta^2} \xi_t,
\end{align*}
where $\xi_1, \xi_2, \ldots \stackrel{iid}\sim \Ncal(0, 1)$, $\rho_\eta = 1 - \nu_\eta^{-\chi}, \chi\in (0,1/2)$. Note that we chose $\chi=1/3$ in our implementation.

\paragraph{Assumptions on the time-series} Our dependence conditions on the time series are formulated in terms of $\beta$-mixing coefficients, which we recall first.
\begin{definition}
    Let $(\Omega,\mathcal{A},P)$ be a probability space and $\mathcal{A}_1,\mathcal{A}_2\subset \mathcal{A}$
    sub-$\sigma$-algebras.
    The \emph{$\beta$-mixing coefficient} is defined as
    $$\beta(\mathcal{F}_1,\mathcal{F}_2)=\delta(P_{\Acal_1}\otimes P_{\Acal_2},P_{\Acal_1\otimes \Acal_2}),$$
    where $\delta$ is the total variation distance and $P_{\Acal_i}$ is the restriction of $P$ to $\Acal_i$.
    For a sequence $X_{i}$ of random variables with common (co)domain, $p<n$ define $$\beta_n^X(p)=\sup_{k\leq n-p}\beta\left(\sigma(X_{1},\ldots,X_{k}),\sigma(X_{k+p},\ldots,X_{n})\right).$$
    We write $\beta_n(p)$ instead of $\beta^X_n(p)$ whenever it is clear from context. 
\end{definition}

Let $X_i\in \R$ be random variables such that there is a $\gamma>2$ satisfying the following conditions:
\begin{enumerate}[label=(S\arabic*), leftmargin=1.3cm]
        \item  $\sup_{i}\E[|X_i|^\gamma]<\infty$
        \item $\sup_n \max_{m\leq n}m^{b}\beta_n(m)<\infty$ with $b \chi >\gamma/(\gamma-2).$
\end{enumerate}

The assumptions require the existence of slightly more than a second moment and a restriction on the serial dependence. The latter requires events to become increasingly independent as they become more separated in time. The dependence has to fade more quickly if the tails of $X_i$ are heavy (i.e., $\gamma$ is small). These conditions are standard in the literature on time series analysis and comparatively mild.

\paragraph{Assumptions on the weights}
Set $w_{t,\eta}(i)=0$ for $t<i$, define
\begin{align*}
    g_{\eta,s}&:\{1, \dots, c\nu_{\eta} \}\to \R, i\mapsto \nu_\eta w_{\lfloor sc\nu_{\eta}\rfloor, \eta}(i)
\end{align*}
and the semi-metric $d^w_\eta$ on $[1/c,1]$ by
$$d^w_\eta(s,t)=\|g_{\eta,s}-g_{\eta,t}\|_{\gamma,\eta}=\left(\frac{1}{\nu_\eta}\sum_{i=1}^{\infty}|\nu_\eta w_{\lfloor sc\nu_{\eta}\rfloor, \eta}(i)-\nu_\eta w_{\lfloor tc\nu_{\eta}\rfloor, \eta}(i)|^{\gamma}\right)^{1/\gamma}.$$


We impose the following conditions on the weights: 
\begin{enumerate}[label=(W\arabic*), leftmargin=1.5cm]
    \item  $\sup_{t,i,\eta}|\nu_{\eta}w_{t, \eta}(i)|<\infty$ 
    \item   There is a constant $C  \in [1, \infty)$ such that $d_\eta^w(s, t) \le C (| s - t| + \nu_\eta^{-1})^{1/C} < \infty$ for all $s,t\in [1/c,1]$.
\end{enumerate}

The first prevents the weights from spiking too heavily, while the second requires the weights to be sufficiently smooth in $s$. The conditions are mild and verified as an example for the EWMA weights in Lemma \ref{lem:weights} below. 

Condition \ref{A5:weights} is a sufficient condition for a set of more technical conditions, which are used in the proofs of our main results below.
Set $\Wcal_\eta = \{g_{\eta,s}\colon s\in [1/c,1]\}$ and let  $N_{\left[\right]}\left( \epsilon,\Wcal_\eta,\|\cdot\|_{\gamma,\eta}\right)$ denote the bracketing number of $\Wcal_\eta$ with respect to the $\|\cdot\|_{\gamma,\eta}$-norm \citep[e.g.,][Definition 2.1.6]{van2023weak}.

\begin{lemma}
    Condition \ref{A5:weights} implies the following:
    \begin{enumerate}[label=(W\arabic*), leftmargin=1.5cm]
        \setcounter{enumi}{2}
        \item  \label{A1:weights} $N_{\left[\right]}\left( \epsilon,\Wcal_\eta,\|\cdot\|_{\gamma,\eta}\right)=\mathcal{O}(\epsilon^{-C})$ for some $C \in (0, \infty)$.
    \item \label{A2:weights} there exists a semi-metric $d^w$ on $[1/c,1]$ such that for all $\delta_\eta\downarrow 0$
    $$\lim_{\nu_\eta\to \infty}\sup_{d^w(s,t)<\delta_\eta}d^w_{\eta}(s,t)=0.$$
    \item \label{A3:weights}$([1/c,1],d^w)$ is totally bounded.
    \end{enumerate}
    The semi-metric $d^w$ can be chosen as $d^w(s,t)=|s-t|$.
\end{lemma}
\begin{proof} 
    Condition \ref{A3:weights} is immediate for $d^w(s,t)=|s-t|$, and 
    \ref{A2:weights} follows from 
    \begin{align*}
        \sup_{|s-t|<\delta_\eta}d_\eta^w(s,t)\le C (\delta_\eta + \nu_\eta^{-1})^{1/C}\to 0.
    \end{align*}
    For \ref{A1:weights}, we distinguish two cases:
    
       \begin{itemize}
        \item Case $\epsilon \ge \nu_\eta^{-1/C}$: Let $s_0 = 1/c$, and $s_1,\ldots,s_N\in [1/c,1 + \eps^C]$ be such that $s_{i}-s_{i - 1}= \epsilon^{C}$, $i = 1, \dots, N$. Then by \ref{A5:weights}, 
        \begin{align*}
            d_\eta^w(s_i,s_{i - 1}) \le C (\epsilon^C + \nu_\eta^{-1})^{1/C} \le 2C (\epsilon + \nu_\eta^{-1/C}) \le 4C\epsilon.
        \end{align*}
        Thus, the brackets $[g_{\eta,s_{i - 1}}, g_{\eta,s_i}]$ have $\|\cdot\|_{\gamma,\eta}$-size at most $4C\epsilon$. Since $N = \lceil (1 + \epsilon^C - 1/c) / \epsilon^C\rceil = O(\epsilon^{-C})$, we have 
        $$N_{\left[\right]}\left( 4C\epsilon,\Wcal_\eta,\|\cdot\|_{\gamma,\eta}\right)=O(\epsilon^{-C}).$$ 
        \item Case $\epsilon < \nu_\eta^{-1/C}$: Due to the rounding operator, the class $\Wcal_\eta$ has at most $c\nu_\eta$ distinct elements. Thus, 
        $$ N\left( 4C\epsilon,\Wcal_\eta,\|\cdot\|_{\gamma,\eta}\right)\le |W_\eta| = c\nu_\eta = O((\nu_\eta^{-1/C})^C)= O(\epsilon^{-C}).$$
       \end{itemize}
       Finally, since $N( 4C\epsilon,\Wcal_\eta,\|\cdot\|_{\gamma,\eta})= O(\epsilon^{-C})$ in both cases, it must hold that 
       \begin{align*}
        N_{\left[\right]}( \epsilon,\Wcal_\eta,\|\cdot\|_{\gamma,\eta})= O\left((\epsilon/4 C)^{-C}\right) = O(\epsilon^{-C}). \tag*{\qedhere}
       \end{align*}
\end{proof}

\begin{lemma}[Verifying assumptions for EWMA weights]\label{lem:weights}
Let 
\begin{align*}
\nu_\eta:=\frac{2-\eta}{\eta},\qquad
w_{t,\eta}(i):=\eta(1-\eta)^{t-i}\mathbf 1_{i\le t }.
\end{align*}
Then for all $K\ge 1$, $\eta\in(0,1)$, $s,t \in [0, 1]$, 
it holds that
\begin{align*}
   \sup_{t, i, \eta} | \nu_\eta w_{t,\eta}(i)| \le 2 \qquad \text{ and }\qquad 
   d_\eta^w(s, t) \le 4K^{1/\gamma}(|s  - t|+\eta)^{1/\gamma}.
\end{align*}
\end{lemma}

\begin{proof}
For the first inequality, note that $\nu_\eta w_{t,\eta}(i)\le \nu_\eta \eta = 2-\eta\le 2$ for all $t,i,\eta$. For the second inequality, assume w.l.o.g.\ $m_t\ge m_s$ and define $\Delta:=\big|m_t-m_s\big|.$
For $i\le m_s$,
\begin{align*}
|g_{\eta,s}(i)-g_{\eta,t}(i)|
=\nu_\eta\eta(1-\eta)^{m_s-i}\big(1-(1-\eta)^{\Delta}\big),
\end{align*}
whence
\begin{align*}
S_1
:=\frac{1}{\nu_\eta}\sum_{i=1}^{m_s}|g_{\eta,s}(i)-g_{\eta,t}(i)|^\gamma
\le \nu_\eta^{\gamma-1}\eta^\gamma
\frac{\big(1-(1-\eta)^\Delta\big)^\gamma}{1-(1-\eta)^\gamma}.
\end{align*}
For $m_s<i\le m_t$,
\begin{align*}
|g_{\eta,s}(i)-g_{\eta,t}(i)|
=\nu_\eta\eta(1-\eta)^{m_t-i},\qquad
S_2:=\frac{1}{\nu_\eta}\sum_{i=m_s+1}^{m_t}|g_{\eta,s}-g_{\eta,t}|^\gamma
=\nu_\eta^{\gamma-1}\eta^\gamma\frac{1-(1-\eta)^{\gamma\Delta}}{1-(1-\eta)^\gamma}.
\end{align*}
Thus
\begin{align*}
\|g_{\eta,s}-g_{\eta,t}\|_{\gamma,\eta}^\gamma
&\le \nu_\eta^{\gamma-1}\eta^\gamma
\frac{\left(1-(1-\eta)^\Delta\right)^\gamma+\left(1-(1-\eta)^{\gamma\Delta}\right)}{1-(1-\eta)^\gamma}.
\end{align*}
Use $1-(1-\eta)^x\le \min\{1,\eta x\}$ for $x\ge 0$, and
$1-(1-\eta)^\gamma\ge \frac{\gamma}{2}\eta$ for small $\eta$, and note
$\nu_\eta^{\gamma-1}\eta^{\gamma-1}=(2-\eta)^{\gamma-1} \le 2^{\gamma - 1}$, giving
\begin{align*}
\|g_{\eta,s}-g_{\eta,t}\|_{\gamma,\eta}^\gamma
\le \frac{2^\gamma}{\gamma} 
\left(\min\{1,(\eta\Delta)^\gamma\}+\min\{1,\gamma\eta\Delta\}\right).
\end{align*}
Since $\Delta, \eta \in [0, 1]$, $\gamma > 1$, we further have $\eta^\gamma \le \eta$ and $\Delta^\gamma \le \Delta$, yielding
\begin{align*}
\|g_{\eta,s}-g_{\eta,t}\|_{\gamma,\eta}^\gamma
&\le 2^\gamma \eta \Delta.
\end{align*}
Finally,
\begin{align*}
\eta\Delta=\eta|m_t-m_s|
\le \eta(c\nu_\eta|t-s|+1)
\le 2K|t-s|+\eta,
\end{align*}
which yields the stated consequences.
\end{proof}

\subsection{Main Results}

Set
\begin{align*}
    \wh \mu_\eta(t) = \sum_{i = 1}^t w_{t, \eta}(i) X_i, 
    \qquad \mu_\eta(t) = \E[\wh \mu_\eta(t)] = \sum_{i = 1}^t w_{t, \eta}(i) \mu(i), 
    \qquad \wh \delta_{\eta}^{*}(t) = \sum_{i = 1}^t w_{t, \eta}(i) V_i(X_i - \wh \mu_{\eta}(i)),
\end{align*}
and define $c_{t, \alpha}^* = \sigma_t^* q_{\alpha}^*$ where $q_{\alpha}^*$ is such that 
$$\Pr\left( \exists t \in (t_1, t_2] \colon |\wh \delta_{\eta}^{*}(t)|/\sigma_t^* >q_{\alpha}^*\right) = 
\Pr\left( \sup_{t \in (t_1, t_2]} |\wh \delta_{\eta}^{*}(t)|/\sigma_t^* >q_{\alpha}^*\right) =\alpha.$$


\begin{theorem}\label{thm:1}
    Assume that \ref{A0:weights}--\ref{A3:weights} and \ref{A1:X}--\ref{A2:X} hold. 
Assume further that $\E[X_i]$ is constant in $i$ and
$\sum_{i=1}^tw_{t,\eta}(i)=1$.
Then,   
\begin{align*}
        \limsup_{\nu_\eta \to \infty}\Pr\left( \exists t \in (t_1, t_2] \colon |\wh \mu_{\eta}(t) - \mu_{\eta}(t)| > c_{t, \alpha}^*\right) = \alpha.
    \end{align*}
\end{theorem}

\begin{theorem}\label{thm:2}
    Assume that \ref{A0:weights}--\ref{A3:weights} and \ref{A1:X}--\ref{A2:X} hold. 
    Assume further 
$$\liminf_{\nu_\eta\to \infty} \min_{t\in [t_0,t_1]}\nu_\eta\var\left[\sum_{i=1}^tw_{t,\eta}(i)V_i\E[\wh \mu_\eta(i)-X_i]\right]=\infty.$$
Then, $\lim_{\nu_\eta\to \infty}\max_{t\in (t_1,t_2]}\sigma_t/\sigma_t^*=0$ and
\begin{align*}
        \limsup_{\nu_\eta \to \infty}\Pr\left( \exists t \in (t_1, t_2] \colon |\wh \mu_{\eta}(t) - \mu_{\eta}(t)| > c_{t, \alpha}^*\right) = 0.
    \end{align*}
\end{theorem}


\subsection{Proofs}

\paragraph{Outline}

The proofs are divided into two parts: consistency of the bootstrap procedure and a comparison of 
the bootstrapped and actual asymptotic variance. 
Both rest on the statistical properties of the transformed multipliers, which are derived last (\cref{sec:transformation}). 
To prove consistency, we will first
assume that the level of the time series is known.
In this setting, the bootstrap 
estimates the actual variance consistently over the whole time-domain (\Cref{lem:eps-of-n}) and the bootstrap is consistent (\Cref{thm:main}). 
The proof of the latter is closely related to Proposition 4.2 of \cite{palm2025centrallimittheoremsnonstationarity}, and we refer to their proofs for further detail. 
Because $\wh \mu_\eta$ estimates the level consistently under the assumptions of \Cref{thm:1}, we obtain our first result.
For general level estimates,
we calculate the asymptotic bias of the bootstrapped variances.
We will show that the bootstrap overestimates the true variance asymptotically (\Cref{lem:asy-var-with-mean-estimate}), which yields \Cref{thm:2}.

\subsubsection{Bootstrap Consistency for Known Level}
Define 
\begin{align*}
    \G_\eta(s)&=\sqrt{\lfloor c\nu_\eta\rfloor}\sum_{i=1}^{\lfloor sc\nu_\eta\rfloor}w_{\lfloor sc\nu_\eta\rfloor,\eta}(i)(X_i-\E[X_i]),
    \qquad 
    \G_\eta^*(s)=\sqrt{\lfloor c\nu_\eta\rfloor}\sum_{i=1}^{\lfloor sc\nu_\eta\rfloor}V_iw_{\lfloor sc\nu_\eta\rfloor,\eta}(i)(X_i-\E[X_i]),
\end{align*}
for $s\in [1/c,1]$.
Note that 
$$\G_\eta(s)=\sqrt{\lfloor sc\nu_\eta \rfloor} (\wh \mu_\eta(sc\nu_\eta)-\mu_\eta(sc\nu_\eta)).$$

\begin{lemma}\label{lem:eps-of-n}
    Assume \ref{A0:weights}--\ref{A3:weights} and \ref{A1:X}--\ref{A2:X} hold.
    Then,
    $$\sup_{s,t\in [1/c,1]}|\cov[\G_\eta^*(s),\G_\eta^*(t)]-\cov[\G_\eta(s),\G_\eta(t)]|\to 0$$
    for $\nu_\eta\to \infty$.
\end{lemma}
\begin{proof}
    Denote $c\nu_\eta=n$ for simplicity. 
    Write $\omega_{n,i}(s)=nw_{\lfloor sn\rfloor,\eta}(i)\ind\{i\leq sn\}$ and 
    note that $\sup_{n,i,s}|\omega_{n,i}(s)|<\infty$ by assumption.
    Then, 
    $$\G_\eta(s)=\frac{1}{\sqrt{n}}\sum_{i=1}^n \omega_{n,i}(s)(X_i-\E[X_i])$$
    and similar for $\G_\eta^*$.
    By the law of total covariances, 
    \begin{align*}
        \cov[\G_\eta^*(s),\G_\eta^*(t)]-\cov[\G_\eta(s),\G_\eta(t)]
        =\frac{1}{n}\sum_{i,j=1}^{n}
        \omega_{n,i}(s)\omega_{n,j}(t)\left[\cov[V_{i},V_{j}]-1\right]\cov[X_{i},X_{j}]
    \end{align*}
    By Theorem 3 of \cite{doukhan2012mixing}, it holds that
    $$|\cov[X_{i},X_{j}]|\lesssim \sup_{n,i}\E[|X_{i}|_\gamma]^{2/\gamma}\beta_n(|i-j|)^{\frac{\gamma-2}{\gamma}}.$$
    

    A straightforward calculation shows 
    $\cov[Z_i,Z_{i+|h|}]=\rho_\eta^{|h|}=(1-\nu_\eta^{-\chi})^{|h|}$ and
    $$|1-\cov[Z_i,Z_{i+h}]|<\eps$$
    for all $h$ such that 
    $(1-\sqrt[h]{1-\eps})^{-1/\chi}<\nu_\eta$.
    In that case, \Cref{lem:transformation-1-covariance} and \Cref{lem:transformation-1-covariance-satisfied} imply that also 
    $$|\cov[V_i,V_{i+h}]-1|\lesssim |\cov[Z_i,Z_{i+h}]-1|\leq \eps.$$
    A first-order Taylor-approximation around $\eps=0$ yields $$(1-\sqrt[h]{1-\eps})^{-1/\chi}\approx |h|^{1/\chi}.$$
    Thus, 
    \begin{align*}
        \sum_{\substack{i,j\leq n\\ |i-j|\leq n^{\chi}}}|(\cov[V_i,V_{i+|h|}]-1)\cov[X_i,X_{i+|h|}]|
        &\lesssim \eps n,
        \\
        \sum_{\substack{i,j\leq n\\ |i-j|> \nu_\eta^{\chi}}}|(\cov[V_i,V_{i+|h|}]-1)\cov[X_i,X_{i+|h|}]|
        &\lesssim 
        n^2 \beta_n^X(\nu_\eta^{\chi})^{(\gamma-2)/\gamma},
    \end{align*}
    for all $\eps>0$ small. 
    Conclude that 
    \begin{align*}
        \frac{1}{n}\sum_{i,j=1}^n |(\cov[V_i,V_j]-1)\cov[X_i,X_j]|
        \lesssim \eps+ n\beta_n^X(\nu_\eta^{\chi})^{(\gamma-2)/\gamma}
        \to \eps,
    \end{align*}
    by \ref{A2:X}.
    Together, this shows
    \begin{align*}
        &\quad \, \limsup_{n\to \infty}\sup_{s,t\in [1/c,1]}|\cov[\G_\eta^*(s),\G_\eta^*(t)]-\cov[\G_\eta(s),\G_\eta(t)]|
        \\
        &\leq
        \sup_{n,i,s}|w_{n,i}(s)|^2\left[\limsup_n\frac{1}{n}\sum_{i,j=1}^{n}\bigl|\left[\cov[V_{i},V_{j}]-1\right]\cov[X_{i},X_{j}]\bigr| \right]
        \\
        &\lesssim \epsilon.
    \end{align*}
    Since this is true for all $\epsilon>0$ small, taking $\epsilon \to 0$ yields the claim.
\end{proof}

\begin{theorem}\label{thm:main}
    Assume \ref{A0:weights}--\ref{A3:weights} and \ref{A1:X}--\ref{A2:X} hold.
    Then, $\G_\eta^*$ is a consistent bootstrap for $\G_\eta$. 
\end{theorem}

\begin{proof}
    We will conclude by Proposition 4.2 of
    \cite{palm2025centrallimittheoremsnonstationarity}.
    By assumption, $V_i$ are identically distributed, independent of $X_i$ with 
    $$\E[V_i]=0,\quad \var[V_i]=1, \quad \E[|T(V_i)|^\gamma]<\infty.$$

    Note that $2<1+(1-2\chi)^{-1}<\gamma$ is equivalent to 
    $\chi<(\gamma-2)/2(\gamma-1)$, hence, 
    $2\gamma(\gamma-1)/(\gamma-2)^2<\gamma/\chi(\gamma-2)<b$. 
    Thus,
    the assumptions of Theorem 3.7 of
    \cite{palm2025centrallimittheoremsnonstationarity}
    with $k_n=c\nu_\eta$, $S=[1/c,1]$ and $w_{n,i}(s)=\nu_\eta w_{\lfloor sc\nu_{\eta}\rfloor, \eta}(i)\ind\{i\leq \lfloor sc\nu_{\eta}\rfloor\} $
    are satisfied.
    Note that the entropy assumption (iii) of Theorem 3.7
    holds since $X_i$ is univariate. 

    By \Cref{cor:mixing-multiplier} and a similar argument 
    as in the proof of \Cref{lem:eps-of-n}, 
    the conditions of Proposition 4.2 of
    \cite{palm2025centrallimittheoremsnonstationarity}
    are satisfied for some $\chi<\alpha<(\gamma-2)/2(\gamma-1)$, which
    yields the claim. 
\end{proof}

\subsubsection{Bootstrap Consistency for Estimated Level}

For $s\in [1/c,1]$, define
\begin{align*}
    \wh \G_\eta^*(s)&=
    \sqrt{\lfloor sc\nu_\eta \rfloor} \wh\delta_\eta^*(sc\nu_\eta)
    =\sqrt{\lfloor c\nu_\eta\rfloor}\sum_{i=1}^{\lfloor sc\nu_\eta\rfloor}V_iw_{\lfloor sc\nu_\eta \rfloor,\eta}(i)(X_i-\wh \mu_\eta(i)),
    \\
    \Delta_\eta(s)
    &=\sqrt{\lfloor c\nu_\eta\rfloor}\sum_{i=1}^{\lfloor sc\nu_\eta\rfloor}V_iw_{\lfloor sc\nu_\eta \rfloor,\eta}(i)\E[X_i-\wh \mu_\eta(i)].
\end{align*}

\begin{lemma}\label{lem:asy-var-with-mean-estimate}
    Assume \ref{A0:weights} and \ref{A1:X}--\ref{A2:X} hold.
    Then,
    \begin{align*}
        \lim_{\nu_\eta\to \infty}\sup_{s\in [1/c,1]}\bigl|\var[\wh \G^*_\eta(s)]-\bigl(\var[ \G_\eta(s)]+\var[\Delta_\eta(s)]\bigr)\bigr|&= 0.
    \end{align*}
\end{lemma}

\begin{proof}
    We first prove 
    \begin{align*}
        \sup_{s\in [1/c,1]}\bigl|\var[\wh \G^*_\eta(s)]-\bigl(\var[ \G_\eta^*(s)]+\var[\wh \G^*_\eta(s)-\G^*_\eta(s)]\bigr)\bigr|&\to 0.
    \end{align*}
    Since 
    \begin{align*}
        \var[\wh \G^*_\eta(s)]
        =\var[ \G_\eta^*(s)]+\var[\wh \G^*_\eta(s)-\G^*_\eta(s)]
        +2\cov[\G_\eta^*(s),\wh \G^*_\eta(s)-\G^*_\eta(s)]
    \end{align*}
    it suffices to show 
    $$\sup_{s\in [1/c,1]}|\cov[\G^*_\eta(s),\wh \G_\eta^*(s)-\G^*_\eta(s)]|\to 0$$
    for $\nu_\eta\to \infty$.

    Again, set $c\nu_\eta=n$ for simplicity.
    We calculate 
    \begin{align*}
        \wh \G^*_\eta(s)-\G_\eta^*(s)
        &= \sqrt{n}\sum_{i=1}^{\lfloor sn\rfloor}V_iw_{\lfloor sn \rfloor ,\eta}(i)(\E[X_i]-\wh \mu_\eta(i)).
    \end{align*}
    Thus,
    \begin{align*}
        \cov[\G^*_\eta(s),\wh \G_\eta^*(s)-\G^*_\eta(s)]
        &=
        n\sum_{i=1}^{\lfloor sn\rfloor}w_{\lfloor sn\rfloor,\eta}(i)w_{\lfloor sn\rfloor,\eta}(j)\cov[V_i,V_j]\cov[X_i,\wh \mu_\eta(j)]
    \end{align*}
    by the law of total covariances.
    For all $i$
    \begin{align*}
        \sum_{j=1}^n|\cov[X_{i},X_{j}]|&\lesssim \sum_{j=1}^n \beta_n(|i-j|)^{\frac{\gamma-2}{\gamma}}
        \lesssim 
        \sum_{j=1}^n |i-j|^{-b \frac{\gamma-2}{\gamma}}
        \lesssim 
        \sum_{h=0}^\infty h^{-b \frac{\gamma-2}{\gamma}}
        <\infty
    \end{align*}
    since $\chi<1$ implies $1<\chi b(\gamma-2)/\gamma<b(\gamma-2)/\gamma$ by \ref{A2:X}.
    Conclude that
 \begin{align*}
        |\cov[X_i,\wh \mu_\eta(j)]|
        &=\left|\sum_{k=1}^jw_{j,\eta}(k)\cov[X_i,X_k]\right|
        \\
        &\leq \frac{1}{n}\sum_{k=1}^jw_{j,\eta}(k)n|\cov[X_i,X_k]|
        \\
        &\lesssim \frac{1}{n}\sup_i\sum_{k=1}^\infty|\cov[X_i,X_k]|
        \\
        &\lesssim \frac{1}{n}.
    \end{align*}
    Combined with  \Cref{cor:mixing-multiplier}, this gives
    \begin{align*}
        \left|\cov[\G^*_\eta(s),\wh \G_\eta^*(s)-\G^*_\eta(s)]\right|
        &=
        \left|n\sum_{i=1}^{\lfloor sn\rfloor}w_{\lfloor sn\rfloor,\eta}(i)w_{\lfloor sn\rfloor,\eta}(j)\cov[V_i,V_j]\cov[X_i,\wh \mu_\eta(j)]\right|
        \lesssim \frac{1}{n^2}\sum_{i,j=1}^n|\cov[V_i,V_j]|
        =o(1).
    \end{align*}
    Thus, 
    $$\sup_{s\in [1/c,1]}|\cov[\G^*_\eta(s),\wh \G_\eta^*(s)-\G^*_\eta(s)]|\to 0$$
    for $\nu_\eta\to \infty$ which concludes the first part of the proof.

    Since $\var[\G_\eta^*(s)]-\var[\G_\eta(s)]$
    converges uniformly to zero by \Cref{lem:eps-of-n},
    it remains to show 
    $$\sup_{s\in [1/c,1]}\bigl|\var[\wh \G^*_\eta(s)-\G^*_\eta(s)]-\var[\Delta_\eta(s)]\bigr|\to 0.$$
    By the law of total covariances
    \begin{align*}
    &\quad \, \var[\wh \G^*_\eta(s)-\G^*_\eta(s)] \\
    &= n\sum_{i=1}^{\lfloor sn\rfloor}w_{\lfloor sn\rfloor,\eta}(i)w_{\lfloor sn\rfloor,\eta}(j)\cov[V_i,V_j]\E\biggr[\bigl(\E[X_i]-\wh \mu_\eta(i)\bigr)\bigl(\E[X_j]-\wh \mu_\eta(j)\bigr)\biggr]
    \\
    &=n\sum_{i=1}^{\lfloor sn\rfloor}w_{\lfloor sn\rfloor,\eta}(i)w_{\lfloor sn\rfloor,\eta}(j)\cov[V_i,V_j]\biggl(\cov[\wh \mu_\eta(i),\wh \mu_\eta(j)]+\E[X_i-\wh \mu_\eta(i)]\E[X_j-\wh \mu_\eta(j)]\biggr)
    \\
    &=\left[n\sum_{i=1}^{\lfloor sn\rfloor}w_{\lfloor sn\rfloor,\eta}(i)w_{\lfloor sn\rfloor,\eta}(j)\cov[V_i,V_j]\cov[\wh \mu_\eta(i),\wh \mu_\eta(j)]\right]+\var\left[\sqrt{n}\sum_{i=1}^nV_iw_{\lfloor sn\rfloor,\eta}\E[X_i-\wh \mu_\eta(i)]\right]
    \\
    &=\left[n\sum_{i=1}^{\lfloor sn\rfloor}w_{\lfloor sn\rfloor,\eta}(i)w_{\lfloor sn\rfloor,\eta}(j)\cov[V_i,V_j]\cov[\wh \mu_\eta(i),\wh \mu_\eta(j)]\right]+\var[\Delta_\eta(s)].
\end{align*}
By a similar argument as before, 
conclude that the first summand converges to zero uniformly in $s\in [1/c,1]$.
This concludes the proof.
\end{proof}

\begin{lemma}\label{lem:multiplier-tightness}
        Let $(W_i)_{i \in \N}, (Y_i)_{i \in \N} $ be two real-valued sequences of random variables with finite second moment, independent from another and $\E[W_i]=0$. Let $w_{i}:S\to \R$ be functions with $\sup_{i,s}|w_i(s)| \le C_w \in (1, \infty)$ for some set $S$.
        Define 
        $$d(s,t)=\left(\frac{1}{n}\sum_{i=1}^n(w_i(s)-w_i(t))^2\right)^{1/2}.$$
        For $\eps>0$ assume that there exist $N_\eps\in \N$ and $s_1,\ldots,s_{N_\eps}$ such that for all $s\in S$ there exists $i$ with $d(s,s_i)\leq \eps$.
        Then,
        \begin{align*}
            \E\left[\sup_{s\in S}\left|\frac{1}{\sqrt{n}}\sum_{i=1}^nw_i(s)W_iY_i\right|\right]
            \lesssim \max_{i=1,\ldots, n}\E[Y_i^2]^{1/2}\left(\sigma_n N_\eps^{1/2}+\eps n^{1/2}\right)
        \end{align*}
        where 
        $\sigma_n^2=n^{-1}\sum_{i,j=1}^n|\cov[W_i,W_j]|$ and the constant only depends on $C_w$ and $\max_{i} \E[W_i^2]$.
    \end{lemma}

    \begin{proof}
    It holds
        \begin{align*}
            \sup_{s\in S}\left|\frac{1}{\sqrt{n}}\sum_{j=1}^nw_i(s_j)W_iY_i\right|
            &\leq 
            \max_{i=1,\ldots, N_\eps}\left|\frac{1}{\sqrt{n}}\sum_{i=1}^nw_i(s)W_iY_i\right|
            +\sup_{d(s,t)\leq \eps}\left|\frac{1}{\sqrt{n}}\sum_{i=1}^n(w_i(s)-w_{i}(t))W_iY_i\right|
        \end{align*}
        We deal with the maximum first. By independence and the law of total covariances
        \begin{align*}
            \sup_{s \in S} \E\left[\left|\frac{1}{\sqrt{n}}\sum_{i=1}^nw_i(s)W_iY_i\right|^2\right] = \sup_{s \in S}  \var\left[\frac{1}{\sqrt{n}}\sum_{i=1}^nw_i(s)W_iY_i\right]
            &\lesssim \max_{i=1,\ldots, n}\E[Y_i^2]\sigma_n^2,
        \end{align*}
        where the constant only depends on $\max_{i,s}|w_i(s)|$.
        The union bound and Markov's inequality yield
        \begin{align*}
            \Pr\left(\max_{1 \le j \le N_\eps}\left|\frac{1}{\sqrt{n}}\sum_{i=1}^nw_i(s_j)W_iY_i\right| > N_\eps^{1/2}\max_{i=1,\ldots, n}\E[Y_i^2]^{1/2}\sigma_n \delta\right) \leq \frac{1}{\delta^2}.
        \end{align*}
        Converting this tail bound to a bound on the expectation yields
        \begin{align*}
            \E\left[\max_{1 \le j \le N_\eps}\left|\frac{1}{\sqrt{n}}\sum_{i=1}^nw_i(s_j)W_iY_i\right|\right]
            &\lesssim N_\eps^{1/2} \max_{i=1,\ldots, n}\E[Y_i^2]^{1/2} \sigma_n.
        \end{align*}

        The Cauchy-Schwarz inequality gives
        \begin{align*}
            \sup_{d(s,t)\leq \eps}\left|\frac{1}{\sqrt{n}}\sum_{i=1}^n(w_i(s)-w_{i}(t))W_iY_i\right|
            &\leq 
            \sup_{d(s,t)\leq \eps}\left(\frac{1}{n}\sum_{i=1}^n(w_i(s)-w_{i}(t))^2\right)^{1/2}
            \left(\sum_{i=1}^nW_i^2Y_i^2\right)^{1/2}
            \\ 
            \leq 
            \eps \left(\sum_{i=1}^nW_i^2Y_i^2\right)^{1/2}.
        \end{align*}
        It follows
        \begin{align*}
            \E\left[\sup_{d(s,t)\leq \eps}\left|\frac{1}{\sqrt{n}}\sum_{i=1}^n(w_i(s)-w_{i}(t))W_iY_i\right|^2\right]
            &\lesssim
            \eps^2 n \max_{i=1,\ldots, n}\E[Y_i^2],
        \end{align*}
        where the constant only depends on 
        $\max_{i} \E[W_i^2]$.
        Combined, this yields
        \begin{align*}
            \E\left[\sup_{s\in S}\bigl|\frac{1}{\sqrt{n}}\sum_{i=1}^nw_i(s)W_iY_i\bigr|\right]
            \lesssim \max_{i=1,\ldots, n}\E[Y_i^2]^{1/2}\left(N_\eps^{1/2} \sigma_n +\eps n^{1/2}\right). \tag*{\qedhere}
        \end{align*}
\end{proof}

\begin{proof}[Proof of \Cref{thm:1}]
    We will first prove that $\mu_\eta(t)$ consistently estimates $\E[X_t]=c$ for $t,\nu_\eta\to \infty$.
    Calculate 
    \begin{align*}
        c-\E[\wh \mu_\eta(t)]
        =c\left(1-\sum_{i=1}^tw_{t,\eta}(i)\right)=0,
    \end{align*}
    and
    \begin{align*}
        \var[\wh \mu_\eta(t)]
        &=\sum_{i,j=1}^tw_{t,\eta}(i)w_{t,\eta}(j)\cov[X_i,X_j]
        =\mathcal{O}(t^{-1}).
    \end{align*}

    Conclude that $$\E[(\E[X_t]-\wh \mu_\eta(t))^2]=\mathcal{O}(t^{-1}).$$
    Next we prove that $\sup_{s\in [1/c,1]}|\wh \G_\eta^*(s) - \G_\eta^*(s)|$
    converges to zero in probability.
    Again, define $c\nu_\eta = n$ and write 
    $\omega_{n,i}(s)=w_{sn,\eta}(i)$.
    Then, 
    \begin{align*}
        \wh \G_\eta^*(s) - \G_\eta^*(s) 
        & = 
        \frac{1}{\sqrt{n}}\sum_{i=1}^n \omega_{n,i}(s)V_i(\E[X_i]-\wh \mu_\eta(i))
    \end{align*}
     \Cref{cor:mixing-multiplier} yields $n^{-1}\sum_{i,j=1}^n|\cov[V_i,V_j]|=\mathcal{O}(n^{\chi})$.
       We first show that the first $n^{\alpha}$ summands are negligible for any $\alpha \in (0, 1)$. Observe that since $\sup_{i, s}|\omega_{n,i}(s)|\lesssim 1$ and $\sup_i \E[|V_i|] \le 2$ by assumption
    \begin{align*}
         \E\left[\sup_{s\in [1/c,1]}\left|\frac{1}{\sqrt{n}}\sum_{i=1}^{n^{\alpha}} \omega_{n,i}(s)V_i(\E[X_i]-\wh \mu_\eta(i))\right|\right] &\lesssim  \E\left[\frac{1}{\sqrt{n}}\sum_{i=1}^{n^{\alpha}} |V_i||\E[X_i]-\wh \mu_\eta(i)| \right] \\
         &\lesssim  n^{-1/2}  \sum_{i=1}^{n^{\alpha}} \E[|\E[X_i]-\wh \mu_\eta(i)|^2]^{1/2} \\
        &\lesssim  n^{-1/2}  \sum_{i=1}^{n^{\alpha}} i^{-1/2} \\
        &\lesssim n^{-1/2} n^{\alpha/2} \to 0.
     \end{align*}

     Next, we deal with the remaining summands.  By \ref{A0:weights}--\ref{A3:weights}, the conditions on the weights $w_{n,i}(s)=w_i(s)$ in 
    \Cref{lem:multiplier-tightness} are satisfied for some $N_\eps\leq \eps^{-C}$.
    Now pick 
    $\eps=\tilde \eps n^{-(1 - \alpha) / 2}$.
    Applying \Cref{lem:multiplier-tightness} with $Y_i=\E[X_i]-\wh \mu_\eta(i)$ yields
    \begin{align*}
        \E\left[\sup_{s\in [1/c,1]}\bigl|\frac{1}{\sqrt{n}}\sum_{i=n^{\alpha}+1}^{n} \omega_{n,i}(s)V_i(\E[X_i]-\wh \mu_\eta(i))\bigr|\right]
        &\lesssim 
        \max_{n^{\alpha}\leq i\leq n}\E[(\E[X_i]-\wh \mu_\eta(i))^2]^{1/2}(n^{\chi/2}N_\eps^{1/2}+\eps n^{1/2})
        \\
        &\lesssim 
        n^{-\alpha/2}(n^{\chi/2}N_\eps^{1/2}+\eps n^{1/2}) \\
          &\lesssim 
        n^{-\alpha/2}(n^{\chi/2} \tilde  n^{C(1 - \alpha) / 4} \eps^{-C/2} + \tilde \eps n^{-(1 - \alpha) / 2} n^{1/2}) \\
        &= n^{-\alpha/2 + \chi / 2 + C(1 - \alpha) / 4} \tilde \eps^{-C/2}+ \tilde \eps \\
        &\to \tilde \eps,
  \end{align*}
    since  $\chi < 1/2$ by assumption and $\alpha \in (0, 1)$ can be taken arbitrarily close to $1$.
    Since this is true for all $\tilde{\eps}>0$, we obtain 
    that the expectation converges to zero.
    Conclude that 
    $$\E\left[\sup_{s\in [1/c,1]}|\wh \G_\eta^*(s) - \G_\eta^*(s)|\right]\to 0.$$
    
    Lastly, 
    $$\sup_{t\in (t_1,t_2]}|\wh \mu_\eta(t)-\mu_\eta(t)|/\sigma_t=\sup_{s\in (1/c,1]}|\G_\eta(s)|/\var[\G_\eta(s)]^{1/2} \quad \text{and} \quad \sup_{s\in (1/c,1]}|\G^*_\eta(s)|/\var[\G^*_\eta(s)]^{1/2}$$
    are asymptotically equivalent by
    \Cref{thm:main} and \Cref{lem:eps-of-n}.
    By the former and \Cref{lem:asy-var-with-mean-estimate},
    $$\sup_{s\in (1/c,1]}|\G^*_\eta(s)|/\var[\G^*_\eta(s)]^{1/2} \quad \text{and} \quad \sup_{s\in (1/c,1]}|\wh \G^*_\eta(s)|/\var[\wh \G^*_\eta(s)]^{1/2}=\sup_{t\in (t_1,t_2]}|\wh \delta^*_\eta(t)|/\sigma_t^*$$
    are asymptotically equivalent.
    Note that $\var[\Delta_\eta(s)]=0$
    since $\E[X_t-\wh \mu_\eta(t)]=0$.
    Combined, we obtain the claim. 
\end{proof}

\begin{proof}[Proof of \Cref{thm:2}]
    We first prove
    $\lim_{\nu_\eta\to \infty}\max_{t\in (t_1,t_2]}\sigma_t/\sigma_t^*=0.$
    Recall $t_i=c_i\nu_\eta$.
    By construction, for all $t\in (t_1,t_2]$ there exists 
    $s\in [1/c,1]$ such that $$\sigma_t^{2}/\sigma_t^{*2}=\var[\G_\eta(s)]/\var[\wh \G^*_\eta(s)].$$
    Thus, it suffices to prove 
        $$\lim_{\nu_\eta\to \infty}\max_{s\in [1/c,1]}\var[\G_\eta(s)]/\var[\wh \G^*_\eta(s)]=0$$
    By \Cref{lem:asy-var-with-mean-estimate}, $\var[\G^*_\eta(s)]$ is asymptotically given by 
    $\var[ \G_\eta(s)]+\var[\Delta_\eta(s)]$, uniformly in $s$.
    Thus, for $\nu_\eta\to \infty$
    $$\var[ \G_\eta(s)]/(\var[ \G_\eta(s)]+\var[\Delta_\eta(s)]),$$
    (hence, $\sigma_t/\sigma_t^{*}$) converges uniformly to zero since $\var[\Delta_\eta(s)]\to \infty$ and $\var[ \G_\eta(s)]$ is bounded uniformly in $s$. 

    Lastly, as in the proof of Corollary 5.2 of \cite{palm2025centrallimittheoremsnonstationarity}, we derive that $\liminf_{\nu_\eta\to \infty}q_\alpha^*>0$.
    From the proof of \Cref{thm:main} resp. Theorem 3.2 of \cite{palm2025centrallimittheoremsnonstationarity}, we derive that
    $\sup_{t\in (t_1,t_2]}|\wh \mu_\eta(t)-\mu_\eta(t)|/\sigma_t$ is bounded in probability. 
    Conclude that $\sup_{t\in (t_1,t_2]}|\wh \mu_\eta(t)-\mu_\eta(t)|/\sigma_t^*$ converges uniformly to zero in probability. Combined, this yields the second claim.
\end{proof}

\subsubsection{Statistical Properties of Transformed Multipliers} \label{sec:transformation}

Recall the construction of the bootstrap weights:
\begin{align}\label{ex:construction}
    V_{t} = T(Z_t), \quad Z_t = \rho_\eta Z_{t - 1} + \sqrt{1 - \rho_\eta^2} \xi_t,
\end{align}
where $\xi_1, \xi_2, \ldots \stackrel{iid}\sim \Ncal(0, 1)$, $\rho_\eta = 1 - \nu_\eta^{-\chi}, \chi\in (0,1/2)$ and $T:\R\to \R$ is some transformation. Note that we chose $\chi=1/3$ in our implementation. 
Assume that there exists $K:\R^2\to \R$ measurable almost everywhere such that 
    \begin{enumerate}[label=(T\arabic*)]
        \item\label{A1:T} $\E[T(Z_1)]=0,\var[T(Z_1)]=1$ and $\E[|T(Z_1)|^{\gamma}]<\infty$ for some $1+(1-2\chi)^{-1}< \gamma.$
        \item\label{A2:T} $|T(x)-T(y)|\leq K(x,y)|x-y|$ almost everywhere. 
        \item\label{A3:T} $\sup_{i,j}\E[K(Z_i,Z_j)^4]<\infty$.
    \end{enumerate}

\begin{remark}\label{lem:variational-distance}
    Let $P,Q$ be two probability measures defined on the same probability space $(\Omega,\mathcal{F})$ with density function $p$ resp. $q$.
    Recall that the total variation distance is given by
    $$\delta(P,Q)=\sup_{A\in \mathcal{F}}|P(A)-Q(A)| = \frac{1}{2}\|p-q\|_{1}.$$
\end{remark}


\begin{corollary}\label{lem:KL}
    Let $(Z_i)_{i\in \N}$ be a stochastic process. 
    Assume all $P_{(Z_i,\dots,Z_{i+k})}$ admit density functions $f_{i,\dots,i+k}$.
    Then, 
    \begin{align*}
        \beta_n(p)
        &= \frac{1}{2}\sup_{k\leq n-p}\|f_{1,\dots,k,k+p,\ldots n}-f_{1,\dots,k}f_{k+p,\dots,n}\|_1.
    \end{align*}
\end{corollary}


\begin{lemma}\label{lem:l1-markov}
    Let $(Z_i)_{i\in \N}$ be a real-valued Markov process such that each $P_{(Z_{k_1},\ldots,Z_{k_l})}$
    admits a density function $f_{k_1,\ldots,k_l}$.
    Then, 
    $$\|f_{1,\ldots,k,k+h,\ldots,n}-f_{1,\ldots,k}f_{k+h,\ldots,n}\|_1= \|f_{k,k+h}-f_{k}f_{k+h}\|_1.$$
\end{lemma}

\begin{proof}
    Markov's property yields
    \begin{align*}
            f_{1,\ldots,k}f_{k+h,\ldots,n}
        &=f_{k+h,\ldots,n|k+h}f_{1,\ldots,k}f_{k+h}
        \\
        &=f_{k+h,\ldots,n|k+h}f_{1,\ldots,k|k}f_kf_{k+h}
        \\
        f_{1,\ldots,k,k+h,\ldots,n}
        &=f_{k+h,\ldots,n|k+h}f_{1,\ldots,k,k+h}
        \\
        &=f_{k+h,\ldots,n|k+h}f_{k,k+h}f_{1,\ldots,k|k}.
    \end{align*}
    In summary, 
    \begin{align*}
        f_{1,\ldots,k,k+h,\ldots,n}-f_{1,\ldots,k}f_{k+h,\ldots,n}
        &= f_{k+h,\ldots,n|k+h}f_{1,\ldots,k|k}(f_{k,k+h}-f_kf_{k+h})
    \end{align*}
    which yields 
    \begin{align*}
        \|f_{1,\ldots,k,k+h,\ldots,n}-f_{1,\ldots,k}f_{k+h,\ldots,n}\|_1
        &= \int |f_{1,\ldots,k,k+h,\ldots,n}-f_{1,\ldots,k}f_{k+h,\ldots,n}|dx_{1,\ldots, n}
        \\
        &=\int \left(\int f_{k+h,\ldots,n|k+h}dx_{k+h+1,\ldots,n}\int f_{1,\ldots,k|k}dx_{1,\ldots k-1}\right)
        |f_{k,k+h}-f_kf_{k+h}|dx_{k,k+h}
        \\
        &= \int |f_{k,k+h}-f_kf_{k+h}|dx_{k,k+h}
        \\
        &=\|f_{k,k+h}-f_kf_{k+h}\|_1
    \end{align*}
    by Fubini's theorem.

\end{proof}

\begin{corollary}\label{cor:alpha-transformation}
    Assume that $T:\R\to \R$ is measurable almost everywhere.
    Then, 
    \begin{align*}
    \beta(\sigma(V_1,\ldots,V_k),\sigma(V_{k+h},\ldots,V_n))
        &\le \frac{1}{2}\left(-\frac{1}{2}\ln\left(1-\cov[Z_k,Z_{k+h}]^2\right)\right)^{\frac{1}{2}},
    \end{align*}
    with $V_i,Z_i$ defined as in \eqref{ex:construction}.
\end{corollary}

\begin{proof}
    Observe 
    \begin{align*}
        \beta(\sigma(V_1,\ldots,V_k),\sigma(V_{k+h},\ldots,V_n))
        &\leq \beta(\sigma(Z_1,\ldots,Z_k),\sigma(Z_{k+h},\ldots,Z_n))
    \end{align*}
    since each $T$
    is measurable almost everywhere, hence, 
    $\sigma(V_{i},\ldots,V_{i+k})\subset \sigma(Z_{i},\ldots,Z_{i+k})$.

    Next,
    $(Z_i)_{i\in \N}$ is a Markov process, hence,
    \begin{align*}
        \beta_p(n)
        &= \frac{1}{2}\sup_{k\leq n-p}\|f_{Z_1,\dots,Z_n}-f_{Z_1,\dots,Z_{k}}f_{Z_{k+p},\dots,Z_n}\|_1 \\
        &= \frac{1}{2}\sup_{k\leq n-p}\|f_{Z_k,Z_{k+p}}-f_{Z_k}f_{Z_{k+p}}\|_1
    \end{align*}
    by \Cref{lem:l1-markov} combined with \Cref{lem:KL}.
    Applying Pinsker's inequality, we obtain 
    \begin{align*}
        \|f_{Z_k,Z_{k+h}}-f_{Z_k}f_{Z_{k+h}}\|_1^2
        &\le \frac{1}{2}D_{KL}(P_{(Z_k,Z_{k+h})}\|P_{Z_k}\otimes P_{Z_{k+h}})
    \end{align*}
    for all $k,h$ where the right-hand side denotes the mutual information.
    Furthermore, the mutual information between Gaussian distributions is given by 
    $$D_{KL}(P_{(Z_k,Z_{k+h})}\|P_{Z_k}\otimes P_{Z_{k+h}})=-\frac{1}{2}\ln\bigl(1-\corr[Z_k,Z_{k+h}]^2\bigr).$$
    Recall $\var[Z_i]=1$, hence, $\corr[Z_k,Z_{k+h}]=\cov[Z_k,Z_{k+h}]$.
    Together, we obtain the claim.
\end{proof}

\begin{corollary}\label{cor:mixing-multiplier}
    For all $\alpha>\chi$ it holds
    \begin{align*}
        \beta_{c\nu_\eta}^V(\nu_\eta^\alpha)=o\left(\nu_\eta^{-\gamma/(\gamma-2)}\right)
        , \qquad \frac{1}{\nu_{\eta}}\sum_{i,j=1}^{c\nu_\eta}|\cov[V_i,V_j]|=\mathcal{O}(\nu_\eta^{\alpha}),
    \end{align*}
     with $V_i,Z_i$ defined as in \eqref{ex:construction}.
\end{corollary}

\begin{proof}
    It holds 
    \begin{align*}
    \beta(\sigma(V_1,\ldots,V_k),\sigma(V_{k+h},\ldots,V_n))
        &\le \frac{1}{2}\left(-\frac{1}{2}\ln\left(1-\cov[Z_k,Z_{k+h}]^2\right)\right)^{\frac{1}{2}}.
    \end{align*}
    by \Cref{cor:alpha-transformation}.
    For $\alpha$ such that $\chi<\alpha$ we obtain
    $$\cov[Z_k,Z_{k+h}]\leq (1-\nu_\eta^{-\chi})^{\nu_\eta^\alpha}=o\left(\nu_\eta^{-\gamma/(\gamma-2)}\right),$$
    for all $n^{\alpha}\leq h$.
    A first-order Taylor approximation of $\ln(1-x^2)$ around zero yields 
    \begin{align*}
       \frac{1}{2}\left(-\frac{1}{2}\ln\left(1-\cov[Z_k,Z_{k+h}]^2\right)\right)^{\frac{1}{2}}
       \approx \frac{1}{\sqrt{2}}\cov[Z_k,Z_{k+h}],
    \end{align*}
    for $\cov[Z_k,Z_{k+h}]$ close to zero.
    Combined we obtain 
    \begin{align*}
        \beta_{c\nu_\eta}^V(\nu_\eta^\alpha)\lesssim 
        \sup_{k\leq \nu_\eta-\nu_\eta^\alpha}\cov[Z_k,Z_{k+\nu_\eta^\alpha}]=o\left(\nu_\eta^{-\gamma/(\gamma-2)}\right).
    \end{align*}

    By Theorem 3 in \cite{doukhan2012mixing},
    \begin{align*}
        \frac{1}{\nu_\eta}\sum_{\substack{i,j\leq c\nu_\eta\\ |i-j|>\nu_\eta^{\alpha}}}|\cov[V_i,V_j]|
        &\lesssim \beta_{c\nu_\eta}^V(\nu_\eta^{\alpha})^{\frac{\gamma-2}{\gamma}} 
        \to 0,
        \qquad
        \frac{1}{\nu_\eta}\sum_{\substack{i,j\leq c\nu_\eta\\ |i-j|\leq\nu_\eta^{\alpha}}}|\cov[V_i,V_j]|
         \lesssim \E[V_1^2] \nu_\eta^{\alpha}
        =\mathcal{O}(\nu_\eta^{\alpha}).
    \end{align*}
\end{proof}

\begin{lemma}\label{lem:transformation-1-covariance}
    Let $X,Y\sim \mathcal{N}(0,1)$ and $T:\R\to \R$.
    Assume that there exists $K:\R^2\to \R$ measurable almost everywhere such that 
    \begin{enumerate}
        \item\label{A1:transformation-1-covariance} $\E[T(Z)]=0$ and $\var[T(Z)]=1$ for $Z=X,Y$.
        \item $|T(x)-T(y)|\leq K(x,y)|x-y|$ almost everywhere. 
        \item $\E[K(X,Y)^4]\leq L$. 
    \end{enumerate}
    Then, 
    \begin{align*}
        |1-\cov[T(X),T(Y)]|
        &\leq \sqrt{3L}|1-\cov[X,Y]|.
    \end{align*}

\end{lemma}

\begin{proof}
    We calculate 
    \begin{align*}
        \E[(T(X)-T(Y))^2]
        &=\E[T(X)^2]+\E[T(Y)^2]-2\E[T(X)T(Y)]
        \\
        &=2(1-\cov[T(X),T(Y)])
    \end{align*}
   Thus,
    \begin{align*}
        2|1-\cov[T(X),T(Y)]|
        &=\E[(T(X)-T(Y))^2]
        \\
        &\leq \E[K(X,Y)^2(X-Y)^2]
        \\
        &\leq \sqrt{\E[K(X,Y)^4]\E[(X-Y)^4]}
        \\
        &=\sqrt{\E[K(X,Y)^4]3\var[X-Y]^2}
        \\
        &=2\sqrt{3\E[K(X,Y)^4]}|1-\cov[X,Y]|
        \\
        &\leq 2\sqrt{3L} |1-\cov[X,Y]|
    \end{align*}
    where the last two equalities follow since 
    $X-Y$ is normally distributed and $\var[X-Y]=2-2\cov[X,Y]$.
\end{proof}

We now show that the assumptions are satisfied for our choice of transformation.

\begin{lemma} \label{lem:transformation-1-covariance-satisfied}
Let $d\ge 8$ and $t_d$ be the CDF of the unit-variance $t$ distribution with $d$ degrees of freedom.
Define $T(x)= t_d^{-1}(\Phi(x))$ and
\begin{align*}
    K(x,y) = 2\sqrt{2\pi} (1+x^2+y^2)^{2}\exp\left(\tfrac{1}{2d}\max\{x^2,y^2\}\right).
\end{align*}
Then:
\begin{enumerate}[label=(\roman*)]
    \item  $|T(x)-T(y)| \le K(x,y) |x-y|$ for all $x,y\in\R$.
    \item If $(Z_1,Z_2)$ is mean-zero, unit-variance bivariate normal, then $\E[K(Z_1,Z_2)^4] \le C < \infty$ for an absolute constant $C$.
\end{enumerate}
\end{lemma}

\begin{proof} \, \\
\begin{enumerate}[label=(\roman*)]
    \item $T=t_d^{-1}\!\circ\Phi$ is $C^1$ and strictly increasing, hence
\begin{align*}
T'(z)=\frac{\phi(z)}{f_t(T(z))},\qquad |T(x)-T(y)|\le \sup_{z\in[x,y]}T'(z)\,|x-y|.
\end{align*}
For the unit-variance $t_d$ density one has
\begin{align*}
f_t(u)= c_d\Big(1+\frac{u^2}{d-2}\Big)^{-\frac{d+1}{2}},\qquad \quad c_d=\frac{\Gamma(\frac{d+1}{2})}{\sqrt{(d -2)\pi}\,\Gamma(\frac{d}{2})}.
\end{align*}
Thus
\begin{align*}
T'(z)=\frac{1}{c_d}\,\phi(z)\left(1+\frac{T(z)^2}{d-2}\right)^{\frac{d+1}{2}}.
\end{align*}
Using  $t_d(T(z))=\Phi(z)$ and crude Mills bounds, one obtains for all $z\in\R$  the pointwise bound
\begin{align*}
1+\frac{T(z)^2}{d-2}\ \le\ C_1\,(1+z^2)^{2/d}\,e^{\,z^2/d},
\end{align*}
with $C_1=(\sqrt{2\pi}A_d^{1/d}/\sqrt{d-2})^2$ and $A_d=c_d d^{(d-1)/2}$. 
For $d\ge 8$, Stirling's formula gives $\frac{1}{c_d}\le \sqrt{2\pi}$ and $C_1\le e$, whence
\begin{align*}
T'(z)\le \frac{1}{c_d}\,\phi(z)\,\Big(C_1\,(1+z^2)^{2/d}e^{\,z^2/d}\Big)^{\frac{d+1}{2}}
\le C_*\,(1+z^2)^{2}e^{\,\frac{z^2}{2d}},
\end{align*}
with $C_*=2\sqrt{2\pi}$.
Therefore, 
\begin{align*}
|T(x)-T(y)|\le \sup_{z\in[x,y]}T'(z)\,|x-y|
\le C_*\,(1+\max\{x^2,y^2\})^{2}\,e^{\,\frac{1}{2d}\max\{x^2,y^2\}}\,|x-y|,
\end{align*}
and since $(1+\max\{a,b\})^2\le (1+a+b)^2$, we obtain (i) with the stated $K$.

\item For $u,v\ge 0$, it holds
\begin{align*}
(1+u+v)^{8}\le 3^{7}\,(1+u^{8}+v^{8}),\qquad
e^{\frac{2}{d}\max\{u,v\}}\le e^{\frac{2}{d}u}+e^{\frac{2}{d}v}.
\end{align*}
Applying these with $u=Z_1^2$, $v=Z_2^2$,
\begin{align*}
K(Z_1,Z_2)^4
&\le C_*^{\,4}\,3^{7}\big(1+Z_1^{16}+Z_2^{16}\big)\,\big(e^{\frac{2}{d}Z_1^2}+e^{\frac{2}{d}Z_2^2}\big).
\end{align*}
The Cauchy-Schwarz inequality  gives for all $i, j \in \{1, 2\}$ that
    $\E[Z_i^{16}e^{\frac{2}{d}Z_j^2}]^2 \le \E[Z_i^{32}] \E[e^{\frac{4}{d}Z_j^2}],$
and recall that $\E [e^{tZ^2}]=(1-2t)^{-1/2}$ for a standard normal with $t<\tfrac 1 2$.
Taking expectations on both sides of the second-to-last inequality, we get, for $d>4$ (hence $2/d<1/2$),
\begin{align*}
\E[K(Z_1,Z_2)^4]
&\le C_*^{\,4}\,3^{7}\left(1+2 \sqrt{\E[Z_1^{32}]}\right) \frac{2}{\sqrt{1-\frac{4}{d}}},
\end{align*}
which proves (ii). \qedhere
\end{enumerate}

\end{proof}

\end{document}